\documentclass[11pt]{article}
\usepackage{amsmath,amsthm,amsfonts,latexsym,graphicx}
\usepackage{amssymb}
\usepackage{fullpage,color}
\usepackage{subcaption,xspace}
\usepackage{caption}
\usepackage{algorithm}
\usepackage[noend]{algpseudocode}
\usepackage{mathdots}

\usepackage{algorithmicx}

\usepackage{tikz}
\usetikzlibrary{decorations.pathreplacing}
\usepackage[linguistics]{forest}
\usepackage{xpatch}
\usepackage{fullpage}
\makeatletter
\xpatchcmd\ALG@step{\arabic{ALG@line}}{\fmtlinenumber{ALG@line}}{}{}
\makeatother 
\let\fmtlinenumber\arabic 

\allowdisplaybreaks

\usepackage{url}
\usepackage[hidelinks]{hyperref}

\hypersetup{ colorlinks=true, citecolor=blue, urlcolor=blue}
\newcommand{\argmin}{\text{argmin}}
\newcommand{\argmax}{\text{argmax}}

\newcommand{\base}{y}
 
\newtheorem{thm}{Theorem}[section]

\newtheorem{lem}[thm]{Lemma}
\newtheorem{prop}[thm]{Proposition}

\newtheorem{defn}[thm]{Definition}
\theoremstyle{remark}
\newtheorem{rk}[thm]{Remark}

\theoremstyle{plain}

\theoremstyle{remark}

\newcommand{\ignore}[1]{}

\title{  Identifying Approximate Minimizers under  Stochastic Uncertainty}

\author{Hessa Al-Thani\thanks{Funded by Graduate Sponsorship Research Award from the Qatar Research and Development Institute.} \and Viswanath Nagarajan \thanks{Department of Industrial and Operations Engineering, University of Michigan, Ann Arbor, USA. Research supported in part by NSF grant CCF-2418495.} }

\date{}

\begin{document} 

\maketitle

\begin{abstract}
    We study a fundamental stochastic selection problem involving  $n$ independent random variables, each of which can be queried     at some cost.  Given a  tolerance level $\delta$, 
     the goal is to find a value that is $\delta$-approximately minimum (or maximum) over all the random variables, at  minimum expected cost. A solution to this problem is an adaptive sequence of queries, where the choice of the next query may depend on previously-observed values.     Two  variants arise, depending on whether the goal is to find a $\delta$-minimum value or a $\delta$-minimizer.    When all query costs are uniform, we provide a $4$-approximation algorithm for both variants.  When query costs are non-uniform, we provide a $5.83$-approximation algorithm for   the $\delta$-minimum value and a $7.47$-approximation for  the $\delta$-minimizer. All our algorithms rely on non-adaptive policies (that perform a fixed sequence of queries), so we also upper bound the corresponding ``adaptivity'' gaps. Our analysis  relates the stopping probabilities in the algorithm and  optimal policies, where a key step is in proving and using certain stochastic dominance properties.  
\end{abstract}

\section{Introduction}

We study a natural stochastic selection problem that involves querying a set of random variables so as to 
identify  their minimum (or maximum)  value   within a desired precision. Consider a car manufacturer who  wants to chose one design from $n$ options so as to  optimize  some attribute (e.g.,  maximum velocity or energy efficiency). Each  option $i$ corresponds to an  attribute value $X_i$ which is uncertain and drawn from  a known probability distribution. It is  possible to  determine the {\em exact} value of   $X_i$ by  further  testing--- but this incurs some cost $c_i$. Identifying the exact minimum (or maximum) value among   the $X_i$s might be too expensive. Instead, our goal is to identify an {\em  approximately}  minimum (or maximum) value, within a prescribed tolerance level. For example, we might be satisfied with  a value (and  corresponding option) that is within  $10$\% of  the true minimum.  The objective is to  minimize the expected cost.  In this paper, we provide the first  constant-factor approximation algorithm for this problem.

Our problem is related to two lines of work:  {\em stochastic combinatorial optimization} and  {\em optimization under explorable uncertainty}. In stochastic combinatorial optimization, a solution   makes  selections incrementally and adaptively (i.e., the next selection can depend on previously observed random outcomes). An optimal solution here may even require exponential  space to describe. Nevertheless, there has been much recent success in obtaining (efficient) approximation algorithms for such problems, see  e.g., \cite{DGV08,GuhaM07,BGLMNR12,GuptaKNR15,GkenosisGHK18,INZ12,JiangLL020,HellersteinKP21,HellersteinLS24}.  
Optimization problems under explorable uncertainty  involve querying   values drawn from known intervals in order to identify a minimizer. Typically, these results  focus on the {\em competitive ratio}, which relates the algorithm's (expected) query cost to the  optimum query-cost in hindsight, see e.g.,~\cite{kahan1991model,CHAPLICK2021,feder2000computing,MSTerlebach2008,MSTmegow2017randomization,erlebach2016query,BampisDELMS21,megow2023set}. In particular, for the problem of finding an {\em exact} minimizer among $n$ intervals,  
 \cite{kahan1991model} obtained a 2-competitive algorithm  in the adversarial setting  and \cite{CHAPLICK2021} obtained a $1.45$-approximation algorithm in the stochastic setting.  
 The problem we study is a significant generalization of the stochastic exact minimizer problem \cite{CHAPLICK2021}.

\subsection{Problem Definition}
\label{sec:prob}
\def\mv{{\sf MIN}}
\def\vv{{\sf VAL}}
\def\smq{\ensuremath{{\sf SMQ}}\xspace}
\def\smqi{\ensuremath{{\sf SMQI}}\xspace}

In the {\em stochastic minimum query} (\smq) problem, there are $n$ independent discrete random variables  $X_1,...,X_n$  that lie in intervals $I_1,...,I_n$ respectively. The random variables (r.v.s)  may be negative.  We assume that each interval is bounded and closed, i.e., $I_j=[\ell_j,r_j]$ for each $j\in [n]$. We also assume (without loss of generality) that each r.v. has non-zero probability at the endpoints of its interval, i.e., $\Pr[X_j=\ell_j]>0$ and $\Pr[X_j=r_j]>0$  for each $j\in [n]$.\footnote{Otherwise, we can just work with a smaller interval representing the same r.v.}  We will use the terms  random variable (r.v.) and interval interchangeably. The exact value of any  r.v. $X_j$ can only be determined by querying it, which incurs some  cost $c_j\ge 0$. Additionally, we are  given a ``precision'' value $\delta\ge 0$, where  the goal is to identify the minimum value over {\em all} r.v.s up to an additive precision of $\delta$. Formally, if $\mv = \min_{j =1}^n X_j$ then we want to find a  deterministic value $\vv$ such that $\mv \le \vv\le \mv+\delta$. Such a  value $\vv$ is called a {\em $\delta$-minimum} value.  The objective in  \smq is to minimize the expected cost of the queried intervals. 
Note that it may be sufficient to probe only a (small) subset of intervals before stopping. 

We also consider a related, but harder, problem where the goal is to {\em identify} some    $\delta$-minimizer    $i^*\in [n]$, i.e., an interval that   satisfies $X_{i^*}\le \mv+\delta$. We refer to this problem as {\em stochastic minimum query for identification} (\smqi).  If a $\delta$-minimum value is found then it also provides a  $\delta$-minimizer (see \S\ref{subsec:prelim}). However, the converse is not true. So, an \smqi  solution  may return an un-queried a  $\delta$-minimizer   $i^*$ without determining   a  $\delta$-minimum value.

Although our formulation above  uses  {\em additive}  precision (we aim to find a value that is at most $\mv+\delta$), we can also handle {\em multiplicative} precision where the goal is to find a value that is  at most $\alpha \cdot \mv$. This just requires  a simple  logarithmic transformation; see Appendix~\ref{app:multiplicative}. We can also handle the goal of finding the {\em maximum} value by working with negated r.v.s $\{-X_i\}_{i=1}^n$.  

Throughout, we  use $N:=[n]=\{1,2,\dots, n\}$ to denote the index set of the r.v.s.

\paragraph{Adaptive and Non-adaptive policies} Any solution to \smq involves querying  r.v.s sequentially until a   $\delta$-minimum value is found. In general, the sequence  of queries may depend on the realizations of previously queried r.v.s. We refer to such solutions as {\em adaptive} policies. Formally, such a solution can be  described  as a decision tree where each node corresponds to the next r.v. to query and the branches out of a node represent the  realization of the queried r.v. 
{\em Non-adaptive} policies are a special class of solutions where the sequence of queries is fixed upfront: the policy then performs queries in this order until a   $\delta$-minimum value is found. A central notion in stochastic optimization is the {\em adaptivity gap}~\cite{DGV08}, which is the worst-case ratio between the optimal non-adaptive value and the optimal adaptive value.  All our algorithms produce non-adaptive policies and hence also bound the adaptivity gap.

\subsection{Results}
Our first result is on the \smq problem with unit costs, for which we provide a $4$-approximation algorithm. Moreover, we achieve this result via a non-adaptive policy, which also proves an upper bound of $4$ on the adaptivity gap. This algorithm relies on combining two natural policies. The first policy simply queries the r.v. with the smallest left-endpoint. The second policy queries the r.v. that maximizes the probability of stopping in the very next step. When used in isolation, both  these policies  have unbounded approximation ratios. However, interleaving the two policies leads to a constant-factor approximation algorithm. 

We also consider the (harder) unit-cost \smqi problem and show that the  same policy leads to a $4$-approximation  algorithm: the only change is in the criterion to stop, which is now more relaxed. While the algorithm is the same as \smq, the analysis for \smqi is significantly more complex due to the new stopping criterion, which allows us to infer a $\delta$-minimizer $i^*$ even when it has not been queried. Specifically, we prove a stochastic dominance property between r.v.s in our algorithm and the optimum (conditioned on the \smq stopping criterion not occurring), and use this  in  relating the \smqi stopping-probability   in the algorithm and the optimum.       

Our next result is for the \smq problem with non-uniform costs. We obtain a constant-factor approximation again, with a slightly worse ratio of $5.83$.  This is based on combining ideas from the unit-cost algorithm with a ``power-of-two'' approach. In particular, the algorithm proceeds in several iterations, where the $i^{th}$ iteration incurs cost roughly $2^i$. In each iteration $i$, the algorithm selects a subset of r.v.s with cost $O(2^i)$ based on the following two criteria (i) smallest left-endpoint and (ii) maximum  probability of stopping in one step. In order to select the r.v.s for criterion (ii) we need to use a PTAS for an appropriate version of the knapsack problem. 

Finally, we consider  the \smqi problem with non-uniform costs.  Directly using the \smq algorithm for \smqi (as in the unit-cost case) does not work here: it leads to a poor approximation ratio. However, a modification of the \smq algorithm works. Specifically, we modify step (i) above: instead of just selecting a prefix of intervals with the smallest 
left-endpoints, we select an ``almost prefix'' set by skipping some expensive intervals.  We show that this approach leads to an approximation ratio of  $7.47$, which is slightly worse than what we obtain for \smq.  The analysis combines aspects of unit-cost \smqi  
and \smq with non-uniform costs.

\subsection{Related Work}\label{subsec:related}
Computing an approximately  minimum or maximum value by querying a set of random variables is a central question in stochastic optimization. Most of the prior works on this topic have   focused on {\em budgeted} variants. Here, one wants to select a subset of queries of total cost within some budget so as to maximize or minimize the value {\em among the queried} r.v.s. 
The results for the minimization and maximization versions are drastically different. A $1-\frac1e$ approximation algorithm for the budgeted max-value problem   follows from results on stochastic submodular maximization~\cite{AsadpourN16}; more complex ``budget'' constraints can also be handled in this setting \cite{AdamczykSW16,GuptaNS17}. These results also bound the adaptivity gap. In addition,  PTASes are known for  non-adaptive and adaptive versions of budgeted max-value~\cite{FuLX18,SegevS21}.  For the budgeted min-value problem, it is known that the adaptivity gap is unbounded and results for the non-adaptive and adaptive versions  are based on entirely different techniques. \cite{goel2010probe} obtained a bi-criteria approximation algorithm for the non-adaptive problem (the queried subset must be fixed upfront) that achieves a $1+\epsilon$ approximation to the optimal value while exceeding the budget by at most an $O(\log\log m)$ factor, where each r.v. takes an integer value in the range $\{0,1,\dots, m\}$. Subsequently, \cite{wang2022probing} studied the adaptive setting (the queried subset may depend on observed realizations) and obtained a $4$-approximation  while exceeding the budget by at most an $O(\log\log m)$ factor. In contrast to these results, the  goal in \smq is to achieve a value close to the true minimum/maximum taken over {\em all} random variables $X_1, X_2, \dots, X_n$ (not just the queried ones). Moreover, we  want to find an  approximately min/max  value with probability one, as opposed to    optimizing   the expected min/max value.  

A different formulation of the minimum-element problem is studied in \cite{Singla18}: this combines the query-cost and the value of the minimum-queried element into a single objective. They obtain an exact algorithm for this setting, which also extends to a wider class of constrained problems.

 Closely related  to our work,   \cite{CHAPLICK2021} studied the \smqi problem with exact precision, i.e., $\delta=0$. In particular, their goal is to identify an {\em interval} that is an  exact minimizer.  
\cite{CHAPLICK2021} obtained a $1.45$-approximation ratio for general query costs. The \smqi problem that we study allows for arbitrary precision $\delta$, and is significantly more complex than the setting in \cite{CHAPLICK2021}. One indication of the difficulty of handling arbitrary $\delta$ is that    the simpler \smq problem   with $\delta=0$ (where we want to find the exact minimum value) admits a  straightforward exact algorithm that  queries by increasing left-endpoint; however,  this algorithm has an unbounded ratio for  \smq with arbitrary  $\delta$ (see \S\ref{sec:valunit} for an  example).

As mentioned earlier, the \smq problem is also  related to 
optimization problems  
under  explorable uncertainty.  
Apart from the   minimum-value problem~\cite{kahan1991model},  various other problems like computing the median~\cite{feder2000computing}, minimum spanning tree~\cite{MSTerlebach2008,MSTmegow2017randomization} and set selection~\cite{erlebach2016query,BampisDELMS21,megow2023set} have been studied in this setting. 
The key difference from our work is that  these results focus on the competitive ratio. In contrast,  we compare to the optimal policy that is limited in the same manner as the algorithm.   We note that there is an $\tilde{\Omega}(n)$ lower bound on the  competitive ratio  for \smq and \smqi; see Appendix~\ref{app:bad-cr}. Our results show that much better (constant) approximation ratios are achievable for \smq and \smqi in the stochastic setting, relative  to an optimal policy.

\subsection{Preliminaries}\label{subsec:prelim}

\paragraph{Stopping rule for \smq.}  Even without querying any interval, we know that the minimum value is at most $R := \min_{i \in N}\{r_i\}$, the minimum right-endpoint. In order to simplify notation,   we incorporate this information using a dummy r.v.  $X_0 = [R,R]$   that is queried at the start of any policy and incurs no cost. 
We now formally define the  condition under which a  policy for \smq is allowed to stop. We will refer to the partial observations at any point in a policy (i.e.,  values of  r.v.s queried so far)  as the {\em state}. Consider any state, given by a subset $S\subseteq N$ of queried r.v.s along with their observations $\{x_i\}_{i\in S}$. The minimum observed value is $\min_{i\in S} x_i$ and the minimum possible value among the un-queried r.v.s is $\min_{j\in N\setminus S} \ell_j$. The stopping criterion is:
\begin{equation}
    \label{eq:smq-stop-rule}
 \min_{i\in S} x_i\quad \le \quad \min_{j\in N\setminus S} \ell_j \,+\, \delta.
 \end{equation}
If this criterion is met then $\vv=\min_{i\in S} x_i$ is guaranteed to satisfy $\mv\le \vv\le \mv+\delta$, where $\mv=\min_{j\in N} X_j$. Also, $\arg \min_{i\in S} x_i$ is a $\delta$-minimizer. On the other hand, if this criterion is not met then there is no value $v$ that guarantees    $\mv\le v\le \mv+\delta$: 
there is a non-zero probability that the minimum value is  $\min_{j\in N\setminus S} \ell_j $ or  $\min_{i\in S} x_i$ (and these values are more than $\delta$ apart). So,
\begin{prop}
    A policy for \smq can stop if and only if criterion~\eqref{eq:smq-stop-rule} holds. 
\end{prop}

The stopping rule for \smqi is described in \S\ref{subsec:unit-smqi}. An \smqi policy can stop either due to the \smq stopping rule (above) or by inferring an un-queried interval $i^*$ as a $\delta$-minimizer.

\paragraph{Adaptivity gap.}  We show that the adaptivity gap for the \smq problem is more than one: so adaptive policies may indeed perform better. This example also builds some intuition for the problem. 
Consider an instance $\mathcal{I}$ with   three intervals as shown in Figure~\ref{fig:eg}. In particular, $X_1 \in \{0, 3, \infty\} $, $X_2 \in \{1, \infty\} $, $X_3 \in \{2,  \infty\} $ and $\delta = 1$. Let  $\Pr(X_1 = 0) = \frac{1}{3}, \Pr(X_1 = 3) = \frac{1}{3}, \Pr(X_1 = \infty) = \frac{1}{3}, \Pr(X_2 = 1  ) = \epsilon,  \Pr(X_2 = \infty  ) = 1- \epsilon,  \Pr(X_3 = 2  ) = 1 - \epsilon, \Pr(X_3 = \infty) = \epsilon$. An adaptive policy is shown in Figure \ref{fig:ad}, which has cost at most $1 + \frac{2}{3} + \frac{ \epsilon}{3} = \frac{5+\epsilon}{3}$. By a case analysis (see Appendix~\ref{apx:na}) the best non-adaptive cost is $\min\left\{ \frac{6-\epsilon}{3} , \frac{5+2\epsilon}{3}\right\}$. Setting $\epsilon=\frac13$, we obtain  an adaptivity gap of $\frac{17}{16}$. We can also modify this instance slightly to get a worse adaptivity gap of   $\frac{12}{11}$.

\begin{figure}[h!]
    \centering
    \begin{minipage}{0.45\textwidth}
        \centering
        \includegraphics[width=0.65\textwidth]{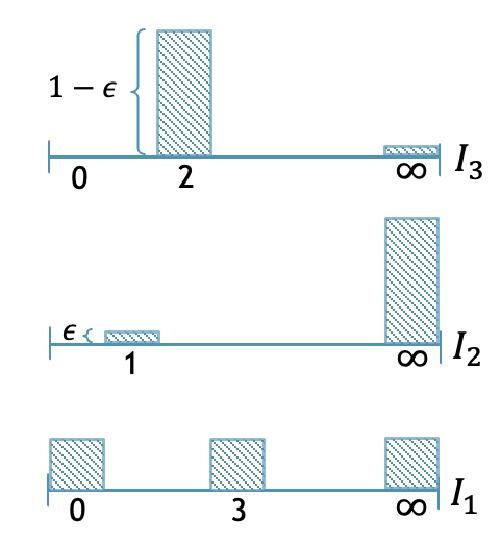}
        \caption{Adaptivity gap instance for \smq. \label{fig:eg}}
        
    \end{minipage}
     \hspace{0.01\textwidth} %
    \begin{minipage}{0.45\textwidth}
        \centering
        \tikzstyle{level 1}=[level distance=20mm, sibling distance=15mm]
        \tikzstyle{level 2}=[level distance=20mm, sibling distance=7mm]
        \tikzstyle{level 3}=[level distance=20mm, sibling distance=5mm]

        \begin{tikzpicture}[grow=right,->]
          \node {$X_1$}
            child {node[draw,fill=black,inner sep=2pt,label=right:{\scriptsize VAL = $0$}] {}
                  edge from parent
                  node[below left] {0}}
            child {node {$X_2$}
                    child{node[draw,fill=black,inner sep=2pt, label=right: {\scriptsize VAL = $1$}]{}
                    edge from parent
                    node[below] {1}}
                    child{node[draw,fill=black,inner sep=2pt, label=right: {\scriptsize VAL = $3$}]{}
                    edge from parent
                    node[above] {$\infty$}}
              edge from parent
              node[below]{3}
            }
            child {node {$X_3$} 
                child {node[draw,fill=black,inner sep=2pt,label={[label distance=-1mm]-15: {\scriptsize VAL = $2$}}] {}
                      edge from parent
                      node[below] {2}} 
                child {node {$X_2$}
                    child {node[draw,fill=black,inner sep=2pt, label=right: {\scriptsize VAL = $1$}]{}
                        edge from parent
                        node[below]{$1$}
                    }
                    child {node[draw,fill=black,inner sep=2pt,label=right: {\scriptsize VAL = $\infty$} ]{}
                        edge from parent
                        node[above]{$\infty$} 
                    }
                      edge from parent
                      node[above] {$\infty$}}
                edge from parent
                node[above left] {$\infty$}           
            };
        \end{tikzpicture}
        \caption{Optimal adaptive policy }
        \label{fig:ad}
    \end{minipage}
\vspace{-1cm}\end{figure}

\paragraph{Fixed threshold problem.} In our analysis, we relate \smq to the following  simpler   problem. Given $n$ r.v.s $\{X_i : i\in N\}$ with costs as before, a {\em fixed} threshold $\theta$ and budget $k$, find a policy having query-cost at most $k$ that maximizes the   probability of observing a realization less than $\theta$. 
A useful property of this fixed threshold problem is that it has adaptivity gap  one;    see Appendix~\ref{app:fixed-t}. 
\begin{prop}\label{prop-fixed-t}
 Consider any instance of the fixed threshold problem.   Let  $V^*$ and $F^*$ denote the maximum success probabilities over adaptive and non-adaptive policies respectively. Then, $V^*=F^*$
\end{prop}

\section{Algorithm for Unit Costs}
\label{sec:valunit}

Before presenting our algorithm, we discuss two simple greedy policies and show why they fail to achieve a good approximation.

\begin{enumerate}
    \item A natural approach is to  select intervals by increasing left-endpoint. Indeed, \cite{kahan1991model} shows that this algorithm is optimal when $\delta = 0$, even in an online setting (with open intervals). 
    Consider the instance with two types of intervals as shown in  Figure \ref{fig:badexample}. The r.v.s $X_1,\dots, X_{n/2}$ are identically distributed with $X_i=0$ w.p. $\frac1n$ and $X_i=n$ otherwise. The remaining r.v.s $X_{n/2+1},\dots, X_{n}$ are identically distributed with $X_i=\frac{\delta}{2}$ w.p. $\frac12$ and $X_i=n$ otherwise. The greedy policy queries r.v.s in the order $1,2,\dots, n$, resulting in an expected cost of $\Omega(n)$ as it can  stop only when it observes a ``low'' realization for some r.v. However, the policy that probes in the reverse order $n,n-1,\dots, 1$ has constant expected cost: the policy can stop upon observing {\em any}  ``low'' realization (even if a  value of $\delta/2$ is observed, it is guaranteed to be within $\delta$ of the true minimum). So the approximation ratio of this greedy policy is  $\Omega(n)$.
   
\begin{figure}[h]
    \centering
    \caption{Bad example for greedy by left-endpoint.}
    \vspace{-0.5cm}
    \[\begin{tikzpicture}[scale=0.4]
    \draw (0,0) -- (9,0); 
    \draw (0,0.3) -- (0,-0.3) node[below] {0};
    \draw (1.5,0.3) -- (1.5,-0.3) node[below] {$\delta/2$};
    \draw (9,0.3) -- (9,-0.3) node[below] {$n$};
    \draw (0,1) -- (9,1) node[anchor = west]{$X_1$};

    \fill (0,1) circle (0.1);
    \fill (9,1) circle (0.1);
    \node at (4.5,2.25) {$\vdots$};
    \draw (0,3)  node[below] { }  -- (9,3) node[anchor = west]{$X_{ n  /2 }$};
    \fill (0,3) circle (0.1);
    \fill (9,3) circle (0.1);
     \draw (1.5,4) -- (9,4) node[anchor =  west]{$X_{n/2 + 1} $};
    \fill (1.5,4) circle (0.1);
    \fill (9,4) circle (0.1);
    \node at (5.5,5.25) {$\vdots$};
         \draw (1.5,6)node[anchor = east ] { }  -- (9,6) node[anchor =  west]{$X_{n}$};
    \fill (1.5,6) circle (0.1);
    \fill (9,6) circle (0.1);
    \draw [decorate,decoration={brace,amplitude=10pt,raise=2pt},xshift = -2 pt, yshift=0pt] (-0.2,1) -- (-0.2,3) node [black,left,xshift=-13pt, yshift=-15pt, align =center] { \tiny \scriptsize $\Pr(X_{i} = 0) = \frac1n$,  \\ \tiny \scriptsize for $i=1,..., \frac{n}2$};

    \draw [decorate,decoration={brace,amplitude=10pt,raise=2pt},xshift = 40 pt, yshift=0pt] (-0.2,4) -- (-0.2,6) node [black,left,xshift=-13pt, yshift=-15pt, align =center] { \tiny \scriptsize $\Pr(X_{i} = \frac{\delta}2) = \frac12$,  \\ \tiny \scriptsize for $i=  \frac{n}2 + 1, ..., n$};
    \end{tikzpicture}
    \label{fig:badexample} \]
\vspace{-1cm}
\end{figure}
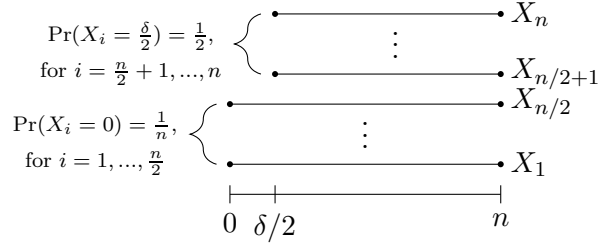

    \item A different greedy policy (based on the instance in Figure \ref{fig:badexample}) is to always  select the interval that maximizes the likelihood of stopping in one step. Now consider another instance with three types of intervals; see Figure~\ref{fig:rightstop}. The r.v. $X_n$ is always $1.4\delta$. The r.v. $X_1$ takes value $0$ w.p. $\frac1{2n}$  and has value $n$   otherwise. The remaining r.v.s $X_2,\dots, X_{n-1}$ are identically distributed with $X_i=\frac{\delta}{2}$ w.p. $\frac1n$ and $X_i=n$ otherwise. 
    As long as $X_1$ is not queried, the probability of stopping (in one step) is as follows:  $\frac1{2n}$ for $X_1$, $\frac1n$ for $X_2,\dots, X_{n-1}$ and zero for $X_n$. So this  greedy policy will query in the order $2,3,\dots, n-1,1,n$ resulting in an $\Omega(n)$ expected cost.  On the other hand, querying the r.v.s $X_1$ and $X_n$ guarantees that the policy can stop. So the optimal cost is at most $2$, implying an $\Omega(n)$ approximation ratio.

    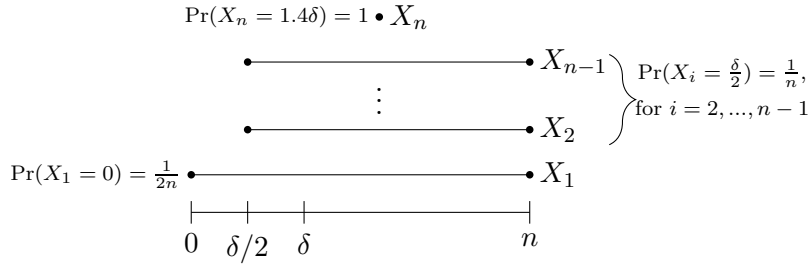
\begin{figure}[!h]
        \centering
    \caption{Bad example for greedy by stopping probability.}        
    \vspace{-0.5cm}
    \[\begin{tikzpicture}[scale=0.5]
    \draw (0,0) -- (9,0); 
    \draw (0,0.3) -- (0,-0.3) node[below] {0};
    \draw (1.5,0.3) -- (1.5,-0.3) node[below] {$\delta/2$};
    \draw (3,0.3) -- (3,-0.3) node[below] {$\delta$};
    \draw (9,0.3) -- (9,-0.3) node[below] {$n$};
    \draw (0,1) node[anchor = east] { \tiny \scriptsize $\Pr(X_{1} = 0) = \frac1{2n}$} --  (9,1) node[anchor = west]{$X_1$};
     \fill (0,1) circle (0.1);
    \fill (9,1) circle (0.1);

        \draw (1.5,2.2)    -- (9,2.2) node[anchor = west]{$X_2$};
    \fill (1.5,2.2) circle (0.1);
    \fill (9,2.2) circle (0.1);
    \node at (5,3.2) {$\vdots$};
    \draw (1.5,4)  node[below] { }  -- (9,4) node[anchor = west]{$X_{n-1}$};
    \fill (1.5,4) circle (0.1);
    \fill (9,4) circle (0.1);

         \draw (5,5.2)node[anchor =  east] { \tiny \scriptsize $\Pr(X_{n} = 1.4\delta ) = 1$ }  -- (5,5.2) node[anchor =  west]{$X_{n}$};
    \fill (5,5.2) circle (0.1);
    \draw [decorate,decoration={brace,amplitude=10pt,mirror,raise=2pt},yshift=0pt, xshift = 22pt] (10.2,1.8) -- (10.2,4.2) node [black,right,xshift= 8 pt , yshift=-13pt, align = center] { \tiny \scriptsize $\Pr(X_{i} = \frac{\delta}2) = \frac1n$, \\ \tiny \scriptsize for $i=2,...,n-1$};

    \end{tikzpicture}
    \label{fig:rightstop} \]
\end{figure}   
\vspace{-1cm}
    
\end{enumerate}

Our approach is to interleave the above two greedy criteria. In particular, each iteration of our algorithm  makes two queries:  the interval with the smallest left-endpoint and the interval that maximizes the probability of stopping in one step.  We will show that this leads to a constant-factor approximation.  
We first re-number   intervals by increasing order of their left-endpoint, i.e.,  $\ell_1 \leq \ell_2\leq \dots \leq \ell_n$.  For each $k\in N$, let $\theta_k := \ell_{k+1} + \delta$. Algorithm \ref{pol:greedy} describes our  algorithm formally.

\begin{algorithm}
\caption{Non-Adaptive Double Greedy}
\label{pol:greedy}
\begin{algorithmic}[1]
\State Let $\ell^* = \min_{i \in N} \ell_i$, $m^* = R :=\min_{i=1}^n r_i$, and $\pi\gets \emptyset$.
\For{$j = 1, \dots, n$}   \Comment{iterations }
\State \label{alg:unit-step-a}  Query interval $j$ (if not already in $\pi$).
\State \label{alg:unit-step-b}  Query interval $b(j) = \argmax_{i \in  N \setminus (\pi \circ j)} \Pr[X_i \le \theta_j] $.
\State Update list $\pi \gets \pi \circ j \circ b(j) $. \Comment{skip $j$ if it was already in $\pi$ }
\State Update $m^* = \min \{m^*, X_j, X_{b(j)}\}$ and $\ell^* = \min_{i \in  N \setminus \pi  }\{\ell_i\}$.
\If{$m^* - \ell^* \leq \delta$}
    stop. \EndIf
\EndFor
\end{algorithmic}
\end{algorithm}

Equivalently, we can view  Algorithm \ref{pol:greedy} as first computing  the permutation  $\pi$ (without querying) and then performing queries in the order given by  $\pi$ until the stopping criterion is met. Note that Algorithm \ref{pol:greedy} is non-adaptive because it uses observations only to determine when to stop. So, our analysis also upper bounds the adaptivity gap.  

We overload notation slightly and use $\pi$ to also denote the non-adaptive policy given in Algorithm \ref{pol:greedy}. Note that each {\em iteration} in this policy involves {\em two} queries. We use $\sigma$ to denote the optimal (adaptive) policy. Let $c_{exp}(\pi)$ and $c_{exp}(\sigma)$ denote the expected number of queries in policies $\pi$ and $\sigma$, respectively. 
The key step in the analysis is to relate the termination probabilities in these two policies, formalized below. 
\begin{lem} \label{lem:unit-main} For any $k\ge 1$, we have 
$$\Pr[\sigma\text{ finishes in }k \text{ queries}]  \leq \Pr[\pi \text{ finishes in 2}k \text{ iterations}].$$    
\end{lem}
We will prove this lemma in the next subsection. First, we complete the analysis using this.
\begin{thm}
\label{thm:main} We have $c_{exp}(\pi) \leq 4\cdot c_{exp}(\sigma)$.    
\end{thm}

\begin{proof}
Let $C_{\sigma}$ denote the random variable that captures the number of queries made by the optimal policy $\sigma$. Similarly, let $C_{\pi}$ denote the number of queries made by our policy. Using Lemma~\ref{lem:unit-main}  and the fact that policy $\pi$  makes two queries in each iteration, for any $k\ge 1$ we have
\begin{equation}
    \Pr[C_{\sigma} \leq k] = \Pr[\sigma\text{ finishes in }k \text{ queries}]  \leq \Pr[\pi \text{ finishes in 2}k \text{ iterations}]  \le \Pr[C_{\pi} \leq 4k]
\end{equation}
Hence,
\begin{align}
    c_{exp}(\sigma) & = \int_{0}^{\infty}{\Pr[C_{\sigma} > t] dt}    = \int_{0}^{\infty} \left( 1 - \Pr[C_{\sigma} \leq t] \right) dt   \geq \int_{0}^{\infty}{(1 - \Pr[C_{\pi} \leq 4t]) dt} \notag \\ 
    & =  \frac{1}{4}\int_{0}^{\infty}{(1 - \Pr[C_{\pi} \leq y]) dy}   = \frac{1}{4} \int_{0}^{\infty}{\Pr[C_{\pi} > y] dy} = \frac{1}{4} c_{exp}(\pi) \label{eq:4t}
\end{align}
The first equality in \eqref{eq:4t} is by a change of variables $y=4t$. 
    \end{proof}

\subsection{Proof of Key Lemma}
We now prove Lemma~\ref{lem:unit-main}. Fix any $k\ge 1$ and define threshold  $\theta:=\theta_k= \ell_{k+1}+\delta$.

Let $T^*\subseteq   N$ denote the optimal solution to the non-adaptive ``fixed threshold'' problem:
\begin{equation}
    \label{eq:t*}
\max_{T\subseteq N, |T|\le k}\quad \Pr\left[ \min_{i\in T}  X_i \le \theta\right].\end{equation}

We then proceed in  two steps, as follows.
$$        \Pr[\sigma\text{ finishes in }k \text{ queries}]  \leq \Pr\left[\min_{i \in T^*} X_i  \le  \theta\right] \leq \Pr[\pi \text{ finishes in }2k \text{ iterations}] 
$$ 
The first inequality is  shown in Lemma \ref{lem:exists}: this uses the fact that  the fixed-threshold problem 
has  adaptivity gap one (Proposition~\ref{prop-fixed-t}). The second inequality is  shown in Lemma \ref{lem:lbalg}: this relies on the greedy criteria used in our algorithm.

\begin{lem}
\label{lem:exists} $\Pr[\sigma\text{ finishes in }k \text{ queries}]  \leq \Pr\left[\min_{i \in T^*} X_i  \le  \theta\right]$.   
\end{lem}
\begin{proof}

Let  $\sigma^k$ denote the optimal policy truncated after $k$  queries: so the cost of $\sigma^k$ is always at most $k$. Let $L(\sigma^k)= \min_{i \in N \setminus {\sigma^k}} \{\ell_i\}$ denote the smallest un-queried left-endpoint at the end of $\sigma^k$; this is a random value because $\sigma^k$ is an adaptive policy. Then, 
\begin{align}
    \Pr[\sigma\text{ finishes in }k \text{ queries}] &= \Pr\left[ \min_{i \in \sigma^k} X_i \,\le \, L(\sigma^k)+ \delta \right] \label{eq:def}\\
    &\leq \Pr \left[ \min_{i \in \sigma^k} X_i  \,\le\, \ell_{k+1} + \delta \right]  = \Pr \left[ \min_{i \in \sigma^k} X_i  \,\le\, \theta\right] \label{eq:lk}\\
    & \le  \Pr\left[\min_{i \in T^*} X_i  \le  \theta\right]
    \label{eq:nonadap} 
\end{align}
Equality \eqref{eq:def} is  by the stopping criterion for \smq. The inequality in \eqref{eq:lk} uses   the observation that after {\em any} $k$  queries,  the smallest un-queried left-endpoint must be at most  $\ell_{k+1}$: so $L(\sigma^k)\le \ell_{k+1}$ always. The equality in \eqref{eq:lk} is by definition of the threshold $\theta$. Inequality  \eqref{eq:nonadap} follows from  Proposition~\ref{prop-fixed-t}: we view  $\sigma^k$ is a feasible adaptive policy for the fixed-threshold problem and $T^*$ is the optimal non-adaptive policy. 
 \end{proof}

\noindent
\begin{lem} 
\label{lem:lbalg} $\Pr\left[\min_{i \in T^*} X_i  \le  \theta\right] \leq \Pr[\pi \text{ finishes in }2k \text{ iterations}] $. \end{lem}

\begin{proof} Recall that each iteration $j$ of Algorithm~\ref{pol:greedy} selects two intervals: $j$ in Step~\ref{alg:unit-step-a} and $b(j)$ in Step~\ref{alg:unit-step-b}. 
Let $B = \{b(1),\dots, b(2k)\}$ be the set of intervals chosen by our policy $\pi$ in Step~\ref{alg:unit-step-b} of the first $2k$ iterations. We partition $B$ into  $B'=\{b(1),\dots, b(k)\}$ and $B''=\{b(k+1),\dots, b(2k)\}$. 
Let $d^* = \argmin_{d \in T^* \setminus B} \Pr(X_d > \theta_k)= \argmax_{d \in T^* \setminus B} \Pr(X_d \le \theta_k)$.
\begin{align}
   \Pr\left[\min_{i \in T^*} X_i \, >\, \theta_k\right]  & =  \prod_{i \in T^*} {\Pr[X_i > \theta_k]} \notag \\
   & =  \prod_{i \in T^* \cap B}  \Pr[X_i > \theta_{k}]  \cdot \prod_{i \in  T^* \setminus B} \Pr[X_i > \theta_{k}]  \notag \\
& \ge  \prod_{i \in T^* \cap B} \Pr[X_i  > \theta_{k}] \,\, \cdot\,\, 
\left({\Pr[X_{d^*} > \theta_{k}]}\right)^{|T^* \setminus B|} \label{eq:dstar} \\
   & \ge  \prod_{i \in T^* \cap B} \Pr[X_i > \theta_{2k}]  \cdot \prod_{i \in  B'' \setminus T^*} \Pr[X_i > \theta_{2k}] \label{eq:ge} \\ 
    & \ge \prod_{i \in  B} \Pr[X_i > \theta_{2k}]  = \Pr\left[\min_{i \in B}  X_i \, >\, \theta_{2k}\right]  \label{eq:ge2}
      \end{align}
\eqref{eq:dstar} follows from the definition of $d^*$. \eqref{eq:ge2} just uses that $T^* \cap B$ and $ B'' \setminus T^*$ are disjoint subsets of $B$. The key step above is \eqref{eq:ge}, which we prove using two cases:
\begin{itemize}
    \item Suppose that $T^* \setminus  B = \emptyset$. Then, using   $\theta_k \leq \theta_{2k}$ we obtain $\Pr[X_i > \theta_{k}]\ge \Pr[X_i > \theta_{2k}]$,  which proves  \eqref{eq:ge} for this case.
    \item Suppose that $T^* \setminus  B \ne \emptyset$. In this case, $d^*$ is well-defined. We now claim that: 
    \begin{equation}\label{eq:claim-index-dom}
        \mbox{  For each $j=k+1,\dots, 2k$,     either } b(j)\in T^* \mbox{ or } \Pr[X_{b(j)} > \theta_{2k}] \leq \Pr[X_{d^*} > \theta_{k}]. 
    \end{equation}

Indeed, consider any such $j$ and suppose that $b(j)\not\in T^*$. As $d^*$ is a valid choice for $b(j)$, the greedy rule implies:
$$\Pr\left[X_{d^*} > \theta_j\right] \ge \Pr\left[X_{b(j)} > \theta_j\right]. $$
Further, using the fact that $\theta_k \le \theta_j\le \theta_{2k}$, we get 
$$\Pr\left[X_{d^*} > \theta_k\right] \ge\Pr\left[X_{d^*} > \theta_j\right] \ge \Pr\left[X_{b(j)} > \theta_j\right] \ge \Pr\left[X_{b(j)} > \theta_{2k}\right] , $$
which proves \eqref{eq:claim-index-dom}. Let $h$ denote the number of iterations $j\in \{k+1,\dots, 2k\}$  where $b(j)\not\in T^*$. Note that $h=|B''| - |T^*\cap B''|= k - |T^*\cap B''| \ge |T^*\setminus B|$, where we used $|T^*|=k$. Using \eqref{eq:claim-index-dom}, it follows that $\Pr[X_{i} > \theta_{2k}] \leq \Pr[X_{d^*} > \theta_{k}]$ for all $i\in B''\setminus T^*$. Hence,
\[
 \Pr[X_{d^*} > \theta_{k}]^{|T^*\setminus B|}  \ge \Pr[X_{d^*} > \theta_{k}]^h  \ge  \prod_{i\in B'' \setminus T^*} \Pr[X_{i} > \theta_{2k}], 
\]
Combined with the fact that $\theta_k \leq \theta_{2k}$  (as before), we obtain  \eqref{eq:ge}.
\end{itemize}

We are now ready to complete the proof. Using the \smq stopping criterion and the fact that $\pi$ queries all the intervals in $B$ within $2k$ iterations, 
$$\Pr[\pi \text{ finishes in }2k \text{ iterations}] \ge \Pr\left[\min_{i \in B}  X_i \, \le\, \theta_{2k}\right]   = 1-\Pr\left[\min_{i \in B}  X_i \, >\, \theta_{2k}\right].$$ 
Combined with \eqref{eq:ge2}, 
$$\Pr[\pi \text{ finishes in }2k \text{ iterations}] \ge 1-  \Pr\left[\min_{i \in T^*} X_i \, >\, \theta_k\right]  =  \Pr\left[\min_{i \in T^*} X_i  \le  \theta\right],$$
where we use the definition $\theta=\theta_k$.
 \end{proof}

\noindent

\subsection{Finding the minimum interval}
\label{subsec:unit-smqi} 
In this section, we consider the \smqi problem,   where the goal is to identify  \emph{an interval} that is guaranteed to be a $\delta$-minimizer. 

Unlike the previous \smq setting (where we find a  $\delta$-minimum value), for \smqi  we just want to identify some interval $i^*\in N$ such that $X_{i^*}\le \mv +\delta$. Recall that $\mv=\min_{i\in N} X_i$. It is important to note that the interval $i^*$ may not have been queried.  It is easy to see that any \smq policy is also feasible to \smqi. Indeed, by the stopping rule \eqref{eq:smq-stop-rule} for \smq, the $\delta$-minimum value returned is always the minimum value of a queried interval: so we also identify $i^*$. However, an \smqi policy may return an interval $i^*$ without querying it. So the optimal value of \smqi may be strictly smaller than \smq. 

\paragraph{Remark:} We note that the optimal values of \smq and \smqi differ by at most the maximum query cost $c_{max}$. As noted above, the optimal \smqi value is at most that of \smq. On the other hand, the optimal \smq value is at most the optimal \smqi value plus the cost to query $i^*$. In the unit-cost setting, $c_{max}=1$ and any policy has expected cost at least $1$: so  the optimal  values of \smq and \smqi are within a factor two of each other.  This immediately implies that Algorithm~\ref{pol:greedy}  is also  an $8$-approximation   for unit-cost \smqi. In the rest of this subsection,  we will prove a stronger result, that  Algorithm~\ref{pol:greedy}  is a $4$-approximation for \smqi. Apart from the improved constant factor, these ideas  will also be helpful for \smqi with general costs. We note that under general costs, the optimal \smq and \smqi values may differ by an arbitrarily large factor because $c_{max}$ is not a lower bound on the optimal value.

\paragraph{Stopping criteria for \smqi} 
Consider any state, given by a subset $S\subseteq N$ of queried r.v.s along with their observations $\{x_i\}_{i\in S}$. There are two conditions under which the \smqi policy can stop.
\begin{itemize}
    \item The first stopping rule is just the one for \smq, Equation~\eqref{eq:smq-stop-rule}. This corresponds to the situation that interval $i^*$ is queried. We restate this rule below for easy reference:  
\begin{equation}
    \label{eq:smqi-rule-1}
 \min_{i\in S} x_i\quad \le \quad \min_{j\in N\setminus S} \ell_j \,+\, \delta.
 \end{equation}
  In this case, we return $i^*=\arg\min_{i\in S} x_i$. We refer to this as the old stopping rule.

\item The second stopping rule handles the situation where an un-queried interval $i^*$ is returned. For any $i\in N$, define the ``almost prefix'' set $P_i:=\{j \in N\setminus i : \ell_j < r_i-\delta\}$. Note that either $P_i$ or $P_i\cup i$ is  a {\em prefix} of $[n]$. (As before, we assume that  intervals are indexed by increasing order of their left-endpoint, i.e.,  $\ell_1 \leq \ell_2\leq \dots \leq  \ell_n$.)     The new rule is:
\begin{equation}
    \label{eq:smqi-rule-2}
\exists \, i\in N  \text{ such that } \,  P_i \subseteq S \text{ and } \min_{j\in P_i} \, x_j \,\ge\, r_i-\delta .
\end{equation}
In other words, there is some  interval $i$ where (1) all intervals $j\ne i$ with left-endpoint $\ell_j<r_i-\delta$ have been queried, and (2) the minimum  value of these r.v.s is at least   $r_i-\delta$. In this case, we return $i^*=i$ (we may not know a $\delta$-minimum value).  We refer to this as the new stopping rule. See Figure~\ref{fig:new-stop} for an example. 
\end{itemize}

\begin{prop}
    A policy for \smqi can stop if and only if either criterion~\eqref{eq:smqi-rule-1} or~\eqref{eq:smqi-rule-2} holds. 
\end{prop}

    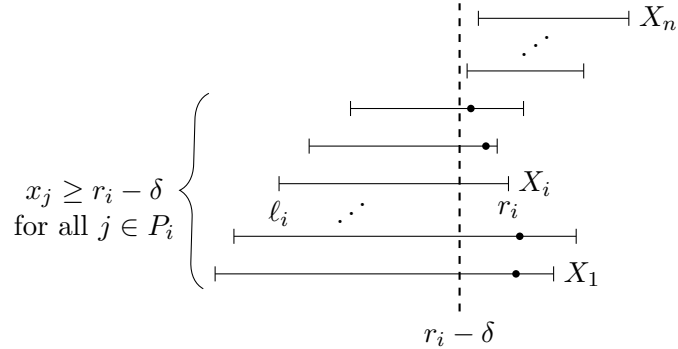
\begin{figure}[!h]
    \centering
    \caption{Illustration of new \smqi stopping criterion. }
    \vspace{-0.5cm}
    \[\begin{tikzpicture}[scale=0.5]

    \draw (0,1)  --  (9,1) node[anchor = west]{$X_1$};
    \draw (0,1.2) -- (0,0.8) node[below] {};
    \draw (9,1.2) -- (9,0.8) node[below] {};
     \fill (8,1) circle (0.1);
    \draw (0.5,2)    -- (9.6,2) node[anchor = west]{};
    \draw (0.5,2.2) -- (0.5,1.8) node[below] {};
    \draw (9.6,2.2) -- (9.6,1.8) node[below] {};
     \fill (8.1,2) circle (0.1);


    \node at (3.6,2.85) {$\iddots$};
 
    \draw (1.7,3.4)  node[below] { }  -- (7.8,3.4) node[anchor = west]{$X_{i}$};
    \draw (1.7,3.6) -- (1.7,3.2) node[below] {$\ell_i$};
    \draw (7.8,3.6) -- (7.8,3.2) node[below] {$r_i$};

    \draw (2.5,4.4)  node[below] { }  -- (7.5,4.4) node[anchor = west]{};
    \draw (2.5,4.2) -- (2.5,4.6) node[below] {};
    \draw (7.5,4.2) -- (7.5,4.6) node[below] {}; 
    \fill (7.2,4.4) circle (0.1);

    \draw (3.6,5.4)  node[below] { }  -- (8.2,5.4) node[anchor = west]{};
    \draw (3.6,5.6) -- (3.6,5.2) node[below] {};
    \draw (8.2,5.6) -- (8.2,5.2) node[below] {};  
    \fill (6.8,5.4) circle (0.1);

    \draw (6.7,6.4)  node[below] { }  -- (9.8,6.4) node[anchor = west]{};
    \draw (6.7,6.6) -- (6.7,6.2) node[below] {};
    \draw (9.8,6.6) -- (9.8,6.2) node[below] {};  
    
    \node at (8.5,7.3) {$\iddots$};
    \draw (7,7.8)  node[below] { }  -- (11,7.8) node[anchor = west]{$X_{n}$};
    \draw (7,8) -- (7,7.6) node[below] {};
    \draw (11,8) -- (11,7.6) node[below] {};
    \draw [decorate,decoration={brace,amplitude=10pt,raise=2pt},yshift=0pt, xshift = 5pt] (-0.35,0.6) -- (-0.2,5.8) node [black,left,xshift= -11 pt , yshift=-44pt, align = center] {   $ {x_j}  \ge r_i - \delta$ \\   for all $j \in P_i$};

\draw[dashed, thick] (6.5, 0) -- (6.5, 8.2);

\node[below, align=center] at (6.5, 0) { $r_i - \delta$};

    \end{tikzpicture}
    \label{fig:new-stop} \]
    \vspace{-1cm}
\end{figure}

Our algorithm for \smqi with unit costs remains the same as for \smq (Algorithm \ref{pol:greedy}). The only difference is in the new stopping criterion (described above). Recall that $\pi$ is the permutation used by our non-adaptive policy. When it is clear from the context, we will also use $\pi$ to denote our \smqi policy that performs queries in the order of $\pi$ until   stopping criteria \eqref{eq:smqi-rule-1} or \eqref{eq:smqi-rule-2} applies.

\def\cG{{\cal G}}
 \def\cE{{\cal E}}
 
\def\cA{{\cal A}}
 \def\cO{{\cal O}}
 
\begin{thm}\label{thm:smqi-unit}
    The non-adaptive policy $\pi$ is a $4$-approximation algorithm for \smqi.
\end{thm}

We now prove this result. Let $\sigma$ denote an optimal adaptive policy for \smqi. For any $k\ge 1$, we will show:
\begin{equation}
\Pr[\sigma\text{ finishes in }k \text{ queries}]  \leq \Pr[\pi \text{ finishes in 2}k \text{ iterations}].  \label{eq:smqi-key} 
\end{equation}
This would suffice to prove the $4$-approximation, exactly as in Theorem~\ref{thm:main}. 

In order to prove \eqref{eq:smqi-key},  we fix some $k\ge 1$. As in the previous proof, let $\theta=\theta_k=\ell_{k+1}+\delta$ and let $T^*$ be defined as in \eqref{eq:t*}. To reduce notation, define the following events. 
$$ \cA_1 \,:\, \text{    our policy $\pi$   finishes within $2k$ iterations due to~\eqref{eq:smqi-rule-1}.}$$
$$ \cA_2 \,:\, \text{    our policy $\pi$   finishes within $2k$ iterations due to~\eqref{eq:smqi-rule-2}.}$$
$$ \cO_1 \,:\, \text{    optimal policy $\sigma$   finishes within $k$ queries due to~\eqref{eq:smqi-rule-1}.}$$
$$ \cO_2 \,:\, \text{    optimal policy $\sigma$   finishes within $k$ queries due to~\eqref{eq:smqi-rule-2}.}$$

\paragraph{Handling the old stopping criterion.}
Let $L$ denote the smallest un-queried left-endpoint at the end of iteration $2k$ in $\pi$.  Note that $L$ is a deterministic value as $\pi$ is a non-adaptive policy. Moreover, $L\ge \ell_{2k+1}$ as $\pi$ would have queried the first $2k$ r.v.s. Let 
$\cG$ be the event that $X_i>L+\delta$ for all intervals $i$  queried by $\pi$ in its first $2k$ iterations. In other words, $\cG$ is precisely the event that stopping criterion \eqref{eq:smqi-rule-1} {\em does not} apply at the end of iteration $2k$ in $\pi$, i.e., $\cG=\neg \cA_1$. By Lemma~\ref{lem:lbalg}, 
\begin{equation*}
\Pr[\neg \cG] =    \Pr\left[\cA_1\right] \ge \Pr\left[\min_{i \in T^*} X_i  \le  \theta\right] .
\end{equation*}

Similarly, let $\cG^*$ be that event that $X_i>\theta$ for all intervals $i$ in the first $k$ queries of $\sigma$. From the proof of Lemma~\ref{lem:exists}, we obtain $\cO_1\subseteq \neg \cG^*$ and 
\begin{equation*}
\Pr[\neg \cG^*] \le   \Pr\left[\min_{i \in T^*} X_i  \le  \theta\right] .
\end{equation*}
Combining the above two inequalities, we have 
\begin{equation}\label{eq:smqi-old}
    \Pr[\neg \cG^*] \le\Pr[\neg \cG].
\end{equation}

\paragraph{Handling the new stopping criterion.} Let $\cG_A$ be the event that $X_j>L+\delta$ for {\em all} r.v.s $j\in N$. Similarly, let $\cG^*_A$ be the event that $X_j>\theta$ for {\em  all}  $j\in N$. Clearly,
\begin{equation}
    \label{eq:smqi-analyze0}
\Pr\left[\cA_2\, |\, \cG\right]  \,=\, \Pr\left[\cA_2\, |\,  \cG_A\right] \quad \text{and} \quad \Pr\left[\cO_2\, |\, \cG^*\right]  \,=\, \Pr\left[\cO_2\, |\,  \cG^*_A\right]. \end{equation}
  We will now prove that
 \begin{equation}\label{eq:smqi-analyze1}
 \Pr\left[\cA_2\, |\, \cG_A\right] \,\ge\,  \Pr\left[\cO_2\, | \, \cG^*_A\right].
\end{equation} 
If $\sigma$ finishes due to~\eqref{eq:smqi-rule-2} in $k$ queries then the almost-prefix set $P_{i^*}\subseteq  [k+1]$: otherwise   $|P_{i^*}|>k$ which  contradicts with the fact that all r.v.s in $P_{i^*} $ must be queried. Let $R=\{i\in N: P_i\subseteq [k+1]\}$ be all such intervals. It now follows that the event $\cO_2$ (which corresponds to policy $\sigma$) is {\em contained in} the  event
\begin{equation}\label{eq:event-E}
    \cE := \bigvee_{i\in R}  \left( \wedge_{j\in P_i} (X_j \ge r_i-\delta)\right) .
    \end{equation}
Note that $\cE$ is independent of the policy: it only depends on the realizations of the r.v.s (and doesn't depend on whether/not an interval  has been queried).

Moreover, our policy $\pi$ queries all the r.v.s in $[2k] \supseteq [k+1]$ within $2k$ iterations. So, for all $i\in R$, the r.v.s in $P_i\subseteq [k+1]$ are queried by $\pi$ in $2k$ iterations. Hence,  event $\cA_2$ (which corresponds to policy $\pi$) {\em contains}   event $\cE$.

Recall that the event $\cG_A$ (resp. $\cG^*_A$) in policy $\pi$ (resp. $\sigma$) means that every  r.v. is more than $L+\delta$ (resp. $\theta$). Also, $\theta \le L+\delta$, which means 
$$\Pr[X_j \ge u| X_j >L+\delta] \ge \Pr[X_j\ge u | X_j >\theta], \quad \forall u\in \mathbb{R} , \forall j\in N. $$
In other words, for any $j\in N$, if  $Y_j$ (resp. $Z_j$) is the  r.v. $X_j$ conditioned on $\cG_A$ (resp. $\cG^*_A$)
then $Y_j$ {\em stochastically dominates} $Z_j$.\footnote{We say that r.v. $Y$ stochastically dominates $Z$ if $\Pr[Y\ge u]\ge \Pr[Z\ge u]$ for all $u\in \mathbb{R}$.} Note also that the r.v.s $Y_j$s (resp. $Z_j$s) are independent. Using the fact that event $\cE$ corresponds to a {\em monotone} function, we obtain:
\begin{lem} \label{lem:event-sd}
Let $\{Y_j:j\in N\}$ and  $\{Z_j:j\in N\}$ be independent r.v.s such that $Y_j$  stochastically dominates $Z_j$ for each $j\in N$. Then, $\Pr[ \cE(Y_1,...,Y_n ) ] \ge  \Pr[\cE(Z_1,...,Z_n)]$ where event $\cE$ is a function of independent r.v.s as defined in \eqref{eq:event-E}.  
\end{lem}
\begin{proof}
It suffices to  prove the following. 
\[ \Pr[ \cE(Y_1,...,Y_h, Z_{h+1}, ...,Z_n) ] \ge  \Pr[\cE(Y_1,...,Y_{h-1}, Z_{h}, ...,Z_n)],\quad \forall h\in [n]. \]
Note that the r.v.s above only differ at position $h$. To keep notation simple, for any $j\in [n]\setminus h$ let $X'_j=Y_j$ if $j<h$ and $X'_j=Z_j$ if $j>h$. So, we need to show $\Pr[\cE(X' , Y_h)] \ge \Pr[\cE(X' , Z_h)]$. 
We  {\em  condition} on the realizations of  the $X'$ r.v.s. For each $j\in [n]\setminus h$ let $t_j$ denote the realization of the  r.v. $X'_j$. 
Having conditioned on these r.v.s, the only randomness is in $Y_h$ and $Z_h$. We will show:
\begin{equation} \label{eq:e-stoch-dom}
\Pr\left[ \cE(X' , Y_h) | X'=t\right]  \ge   \Pr\left[ \cE(X' , Z_h) | X'=t\right] .
\end{equation}
 Using the definition of the event $\cE$ from \eqref{eq:event-E}, let $R(t) = \{i\in R : h\in P_i \text{ and } t_j>r_i-\delta \text{ for all }j\in P_i\setminus h\}$. In other words, $R(t)\subseteq R$ corresponds to those ``clauses'' in  \eqref{eq:event-E} that have not evaluated to true or false based on the realizations $\{ X'_j=t_j : j\in [n]\setminus h\}$. If there is some clause in  \eqref{eq:event-E} that already evaluates to true (based on $t$) then $\cE$ holds regardless of $Y_h$ and $Z_h$. So,  \eqref{eq:e-stoch-dom} holds in this case (both terms are one). Now, we assume that no clause in  \eqref{eq:event-E}  already evaluates to true. We can write  
$$\left\{ \cE(X' , Y_h) | X'=t\right\} \quad =\quad \bigvee_{i \in R(t)  } ( Y_h\ge r_i - \delta) \quad =\quad \left\{ Y_h \ge f \right\}, $$
where $f=  \min_{i\in R(t)} r_i -\delta$ is a deterministic value.\footnote{If $R(t)=\emptyset$ then we set $f=\infty$.} Similarly, we have
$$\left\{ \cE(X' , Z_h) | X'=t\right\} \quad =\quad \left\{ Z_h \ge f \right\}. $$
  Using the fact that $Y_h$ (resp. $Z_h$) is independent of $X'$ and that $Y_h$ stochastically dominates $Z_h$,
$$ \Pr\left[ \cE(X' , Y_h) | X'=t\right] = \Pr[Y_h\ge f] \ge \Pr[Z_h\ge f] = \Pr\left[ \cE(X' , Z_h) | X'=t\right]. $$
This completes the proof of \eqref{eq:e-stoch-dom}. De-conditioning the $X'$ r.v.s, we obtain $\Pr[\cE(X' , Y_h)] \ge \Pr[\cE(X' , Z_h)]$ as desired.
\end{proof}
 
Using Lemma~\ref{lem:event-sd}, we obtain   $\Pr[\cE  | \cG_A] \ge \Pr[\cE | \cG^*_A]$, which proves \eqref{eq:smqi-analyze1}. Combined with \eqref{eq:smqi-analyze0},  
 \begin{equation}\label{eq:smqi-analyze2}
 \Pr\left[\cA_2\, |\, \cG \right] \,\ge\,  \Pr\left[\cO_2\, | \, \cG^* \right].
\end{equation} 

\paragraph{Wrapping up.} We have
\begin{align*}
    \Pr[\cA_1 \vee \cA_2] &=     \Pr[\cA_1] +     \Pr[\cA_2 \wedge \neg \cA_1] = \Pr[\neg \cG] + \Pr[\cA_2 \wedge  \cG]\\
    & = \Pr[\neg \cG] + \Pr[\cA_2 |\cG]\cdot \Pr[\cG] = 1 - \left( 1 -  \Pr[\cA_2 |\cG]\right) \cdot \Pr[\cG]\\
    & \ge 1 - \left( 1 -  \Pr[\cO_2 |\cG^*]\right) \cdot \Pr[\cG^*]\qquad \text{ by \eqref{eq:smqi-old} and \eqref{eq:smqi-analyze2}}\\
    &= \Pr[\neg \cG^*] + \Pr[\cO_2 \wedge  \cG^*]\\
    & \ge \Pr[\cO_1] +     \Pr[\cO_2 \wedge \neg \cO_1] \qquad \text{ using  }\cO_1\subseteq \neg \cG^* \\
    & =    \Pr[\cO_1 \vee \cO_2].
\end{align*}
This completes the proof of \eqref{eq:smqi-key} and the theorem.

\def\i{g}  
\section{Algorithm for  General Costs}

We now consider the \smq problem with non-uniform query costs. We assume (without loss of generality, by scaling) that  costs are at least one, i.e., $\min_{i\in N} c_i\ge 1$. The high-level idea is similar to 
the unit-cost case: interleaving the two greedy criteria of smallest left-endpoint and highest probability of stopping. However, we need to incorporate the costs carefully. To this end, we use an iterative algorithm that in every iteration $\i$, makes  a {\em batch} of queries having total cost about $2^\i$. (In order to optimize the approximation ratio, we use a generic base $\base$ for the exponential costs.)

For any subset $S\subseteq N$, let $c(S):=\sum_{j\in S} c_j$ denote the cost of querying all  intervals in $S$.  Again, we renumber intervals so that $\ell_1\le \ell_2\le \dots \le \ell_n$.

\begin{defn}
For any $\i\ge 0$, let $T_\i$ be the maximal prefix of intervals having cost at most $ {\base}^\i$.
\end{defn}

\begin{algorithm} 
\label{alg:grd-gen}
\caption{Double Greedy for General Cost \label{pol:gengreedy}}
\begin{algorithmic}[1]
\State Let $\ell^* = \min_{j \in N} \ell_j$, $m^* = R:=\min_{j\in N} r_j$,  and $\pi\gets \emptyset$.
\For{$ \i =0, 1, 2, \dots,$} \Comment{iteration}

    \State \label{alg:step-a}  Query intervals $T_\i\setminus \pi$ and update list $\pi \gets \pi \circ T_\i$.
    \State \label{alg:step-theta} Update $\ell^* = \min_{j \in N \setminus \pi }\{\ell_j\}$ and   let threshold $\theta_\i = \ell^* + \delta$. 
    \State \label{alg:step-knap} Compute a $(1,1+\epsilon)$ bicriteria approximate solution  $U_\i$  for: 
    \begin{align}
    p^*_\i = \min_{T \subseteq N \setminus \pi } \left\{\Pr\left[ \min_{j \in T} X_j >  \theta_\i \right] : c(T) \leq {\base}^\i\right\}. \tag{KP}\label{eq:algknap}
    \end{align}
    \State \label{alg:step-b}  Query intervals $U_\i$ and update list $\pi \gets \pi \circ  U_\i$.
    \State Update $\ell^* = \min_{j \in N \setminus \pi }\{\ell_j\}$ and  $m^*=\min\left\{m^*, \min_{j\in T_\i\cup U_\i} X_j\right\}$. 
    \If{$m^* - \ell^* \leq \delta$}
    stop. \EndIf
\EndFor
\end{algorithmic}
\end{algorithm}

The complete algorithm is given in Algorithm~\ref{pol:gengreedy}. 
The optimization problem \eqref{eq:algknap} solved in Step~\ref{alg:step-knap} is a variant of the classic knapsack problem: in Theorem \ref{thm:knapsack} (see Appendix~\ref{app:knapsack}) we provide a  $(1,1+\epsilon)$ bicriteria approximation algorithm for \eqref{eq:algknap} for any constant $\epsilon>0$. In particular, this ensures that $c(U_\i) \leq \base^\i(1+\epsilon)$ and 
$$ \Pr\left[\min_{j \in U_\i} X_j  > \theta_\i\right] \leq p^*_\i .$$ Note that the left-hand-side above equals $\prod_{j \in U_\i}\Pr[ X_j  > \theta_\i]$ as all r.v.s are independent. 

Furthermore, just like Algorithm \ref{pol:greedy}, we can view   Algorithm \ref{pol:gengreedy} as first computing the permutation $\pi$ (without querying) and then performing  queries in that order until the stopping criterion. So, our algorithm is a non-adaptive policy and our analysis also upper-bounds the adaptivity gap. 
 
\subsection{Analysis}
 We use  $\sigma$ to denote the optimal (adaptive) policy and $\pi$ to denote our non-adaptive policy. 
 
 \begin{defn}
\label{def:ai}
    For any $\i \ge 0$, let $ o_\i := \Pr[ \sigma \text{ does not finish by cost }\base^\i]$. 
    Similarly, for our policy we define $v_\i := \Pr[\pi \text{ does not finish by iteration }\i]$.   
We also define $\sigma_\i$ to be the optimal policy truncated at cost $\base^\i$, i.e., the total cost of queried intervals is always at most $\base^\i$. Similarly, we define $\pi_\i$ to be our policy truncated at the end of iteration $\i$.
\end{defn}

The key part of the analysis lies in relating the non-stopping probabilities $o_\i$ and $a_\i$ in the optimal and algorithmic policies: see Lemma~\ref{lem:aioi}. Our first lemma  bounds the (worst-case) cost incurred in $\i$ iterations of our policy.

\begin{lem}
\label{lem:grdup}The cost of our policy until the end of iteration $\i$  is
   \[c(\pi_\i)  \leq (1 + \epsilon) \left(1 + \dfrac{\base}{\base - 1}  \right)\base^\i. \]
\end{lem}
\begin{proof}
We    handle separately the costs of intervals queried in Steps \ref{alg:step-a}   and \ref{alg:step-b}. The total cost incurred in Step~\ref{alg:step-a} of the first $\i$ iterations is $c(T_\i)\le \base^\i$: this uses  $\cup_{k=0}^\i T_k  = T_\i $ because  $T_\i$ are prefixes. The total cost due to Step~\ref{alg:step-b} can be bounded using a geometric series:  
$$\sum_{k=0}^\i c(U_k) \le (1+\epsilon) \sum_{k=0}^\i \base^k=(1 + \epsilon) \cdot \dfrac{\base^{\i+1} - 1}{\base - 1}.$$
The inequality above is by the cost guarantee for \eqref{eq:algknap}. The lemma now follows. 
\end{proof}

 \begin{lem} For all $\i \ge 0$, we have $v_\i \le o_\i$. 
 \label{lem:aioi}
 \end{lem}
\begin{proof}
Recall that  $\sigma_\i$ denotes the optimal policy truncated at cost $\base^\i$. We let $L(\sigma_\i) = \min_{j \in  N \setminus {\sigma_\i} }{\{\ell_j\}}$ be the smallest un-queried left-endpoint: this is a random value as $\sigma_\i$ is adaptive. In the algorithm, consider iteration $\i$ and let $L(T_{\i})= \min_{j \in N\setminus T_{\i}}\{\ell_j\}$; note that the threshold $\theta_\i\ge L(T_{\i}) +\delta$ in Step~\ref{alg:step-theta}. Let   $\pi'=\pi_{\i-1}\circ T_\i$ denote the list after Step~\ref{alg:step-a} in iteration $\i$. Note   that the optimization in \eqref{eq:algknap} of iteration $\i$ is over $T\subseteq N\setminus \pi'$, which yields $U_\i$. Also, $\pi_\i=\pi'\cup U_\i$.
\begin{align}
   o_\i &= \Pr \left[\sigma \text{ does not finish within cost }\base^\i \right]\notag \\
    &= \Pr \left[ \min_{j \in \sigma_\i} X_j  > L(\sigma_\i) + \delta \right]  \ge \Pr \left[ \min_{j \in \sigma_\i} X_j  > L(T_\i) + \delta \right] \label{eq:li}  \\ 
    &\ge \Pr \left[ \min_{j \in \sigma_\i} X_j  > \theta_\i \right]  = 1- \Pr \left[ \min_{j \in \sigma_\i} X_j  \le \theta_\i \right]  \label{eq:opt-theta}  \\
    & \ge 1-  \max_{T \subseteq N, c(T) \leq \base^\i}{\Pr\left[ \min_{j \in {T}} X_j  \le \theta_\i \right]}  = \min_{T \subseteq N, c(T) \leq \base^\i} {\Pr\left[ \min_{j \in {T}} X_j  > \theta_\i \right]} \label{eq:statpol}\\
    & = \min_{T \subseteq N, c(T) \leq \base^\i}{\prod_{j\in T} \Pr\left[ X_j  > \theta_\i \right]} 
     \ge  \prod_{j \in \pi'}\Pr[ X_j  > \theta_\i ] \cdot \min_{T \subseteq N \setminus \pi', c(T) \leq \base^\i}
        \prod_{j\in T} \Pr[ X_j  >  \theta_\i ] \label{eq:stat} \\ 
    & =  \prod_{j \in \pi'}\Pr[ X_j  > \theta_\i ]  \cdot p_\i^* =  \Pr\left[ \min_{j \in \pi'} X_j > \theta_\i \right]  \cdot p_\i^* \label{eq:p} \\
    & \ge  {\Pr\left[ \min_{j \in \pi'} X_j > \theta_\i \right]} \cdot {\Pr\left[ \min_{j \in {U_\i}} X_j  > \theta_\i \right]}  =   {\Pr\left[ \min_{j \in \pi_{\i}} X_j  > \theta_\i \right]} \ge v_\i \label{eq:knap} 
\end{align}
The equality in \eqref{eq:li} is given by the definition of $L(\sigma_\i)$ and the stopping rule. The inequality in \eqref{eq:li} uses the fact that   $L(\sigma_\i)  \leq  L(T_\i)$ always, which in turn is because $\sigma_\i$ has cost at most $\base^\i$ and $T_\i$ is the maximal prefix within this cost.  The inequality in \eqref{eq:opt-theta} uses $\theta_\i\ge L(T_\i)+\delta$. The inequality in \eqref{eq:statpol} is by Proposition~\ref{prop-fixed-t}: we view $\sigma_\i$ as a feasible adaptive policy for the fixed-threshold problem with threshold $\theta_\i$ and budget $y^\i$. The equality in \eqref{eq:stat} follows from independence of the random variables. The first equality in \eqref{eq:p} uses the definition of $p_\i^*$ from \eqref{eq:algknap} and independence. The first inequality in \eqref{eq:knap} uses the choice of $U_\i$ and   Theorem~\ref{thm:knapsack}. The   equality in \eqref{eq:knap} is by $\pi_\i = \pi' \cup U_\i$. To see the last inequality in \eqref{eq:knap}, note that   if $\min_{j \in \pi_{i}}\{X_j\} \le \theta_\i$ then $\pi$   finishes by iteration $\i$. 
\end{proof}

In Lemma \ref{lem:optlb} we lower bound the expected cost of the optimal policy.  Let $c_{exp}(\pi)$ and $c_{exp}(\sigma)$ denote the expected cost of our greedy policy and the optimal policy, respectively.

\begin{lem}
\label{lem:optlb}
For any base   $y \ge 1$, we have $\sum_{\i\ge 0}{\base^\i \cdot o_\i} \leq  \frac{\base }{\base-1} c_{exp}(\sigma) - \frac1{y-1}$.   
\end{lem}
\begin{proof}

    Let $Z$ denote the random variable that represents the cost of the optimal policy $\sigma$: so $c_{exp}(\sigma)=\mathbb{E}[Z]$. Let $\mathbf{1}(Z > \base^\i)$ be the indicator variable for when $Z> \base^\i$; so    $\mathbb{E}[I(Z > \base^\i)]= o_\i$. We now show that:
       \begin{align}
         \sum_{\i\ge 0}{\base^\i \cdot  \mathbf{1}(Z > \base^\i)} &  \le \frac{\base }{\base - 1}Z -\frac1{y-1} \label{eq:sum} 
    \end{align}
To see this, suppose that  $ \base^k < Z \le \base^{k+ 1}$ for some integer $k\ge 0$. Then the left-hand-side of 
\eqref{eq:sum} equals
\[\sum_{\i = 0}^{k}{\base^\i} = \frac{y^{k+1} -1}{y-1} \le    Z \frac{y}{y-1}  - \frac1{y-1} ,\] 
which proves \eqref{eq:sum}. Taking the expectation of \eqref{eq:sum} proves the lemma.    
\end{proof}

\begin{thm}
\label{thm:gencost}

    There is a $( 3+ 2\sqrt{2} + \epsilon)$-approximation for the \smq problem with general costs.  
\end{thm}

\begin{proof} By Lemma~\ref{lem:grdup}, we have   $c_{exp}(\pi)  \leq (1 + \epsilon)  \left(1+ \frac{\base}{\base-1}  \right) \sum_{\i \ge 1}{{\base}^\i(v_{\i-1} - v_\i)}$. Now, 
    \begin{align}
        \sum_{\i \ge 1} {\base}^\i(v_{\i-1} - v_\i) & = v_0 + (y-1) \sum_{\i \ge 0} {\base}^\i v_{\i}  \le 1 + (y-1) \sum_{\i \ge 0} {\base}^\i o_{\i} \label{eq:bylem}\\
        &  \le 1+(y-1) \sum_{\i \ge 0} {\base}^\i o_{\i}  \leq 1+ y\cdot c_{exp}(\sigma) - 1 = y\cdot c_{exp}(\sigma) \label{eq:bylem3}
    \end{align}
The inequality in \eqref{eq:bylem} is by Lemma \ref{lem:grdup} and $v_0=1$.  The first inequality in  \eqref{eq:bylem3} uses $y\ge1$ and the second inequality is by Lemma \ref{lem:optlb}. 

Hence, we obtain $c_{exp}(\pi)  \leq (1 + \epsilon)  y\cdot \left(1+ \frac{\base}{\base-1}   \right) \cdot  c_{exp}(\sigma)$. Now, optimizing for $y$, we obtain the stated approximation ratio.
\end{proof}

\section{\smqi under Non-uniform Costs}
We now consider the (harder) problem of identifying a $\delta$-minimum interval. Recall that the goal here is  to identify some interval $i^*\in N$ such that $X_{i^*}\le \mv +\delta$ where  $\mv=\min_{i\in N} X_i$. The interval $i^*$ may not have been queried by the policy. Unlike the unit-cost case, we can no longer rely on the \smq algorithm itself (see the example below). 

\paragraph{Bad example for the \smq policy.} Consider an instance with the following r.v.s.
\begin{itemize}
    \item $X_1$ is distributed over the interval $[0,1.5\,\delta]$. (The exact distribution is irrelevant.)
    \item $X_2$ has    $\Pr[X_2=0.3\,\delta] =\frac1{n^2} $ and $\Pr[X_2=2\,\delta] =1-\frac1{n^2}$.
    \item $X_3,\dots, X_n$ are identically distributed with $\Pr[X_i=0.7\,\delta] = \frac1n $ and $\Pr[X_i=1.5\,\delta] = 1-\frac1n $.
\end{itemize}
The cost $c_1=n\gg 1$ and all other costs are unit. The \smq policy from Algorithm~\ref{pol:gengreedy} will first select at least $\Omega(n)$ r.v.s among $X_3,\dots, X_n$ because these will optimize \eqref{eq:algknap}. Crucially, the policy will not query $X_1$ or $X_2$ for a long time. Consequently, the expected cost of this policy is $\Omega(n)$. On the other hand, an optimal policy just  queries $X_2$: if $X_2=2\,
\delta$ then it returns $i^*=1$ (stopping rule~\eqref{eq:smqi-rule-2} applies); if $X_2=0.3\,\delta$ then it returns $i^*=2$ (stopping rule~\eqref{eq:smqi-rule-1} applies). So the \smq algorithm has an $\Omega(n)$ approximation ratio when applied directly to \smqi.

This example shows that in the presence of non-uniform  costs,  additional work is needed to handle the new stopping criterion ~\eqref{eq:smqi-rule-2}. In particular, we need to skip  expensive intervals while querying in the order of left-endpoints. The \smqi algorithm for general costs has the same high-level structure as the one for \smq  (Algorithm~\ref{pol:gengreedy}). The only change is in Step~\ref{alg:step-a} where we modify the queried set $T_\i$ 
by skipping some expensive intervals. We first define these new ``almost prefix'' sets that will be used in  Step~\ref{alg:step-a}.

\def\t{S}
\begin{defn}
    For any iteration $\i\ge 0$,  an interval $j\in N$ is called {\bf $\i$-big} if its cost $c_j>\base^\i$ (otherwise, $j$ is called $\i$-small).  
\end{defn}
Recall that the intervals are numbered according to their left-endpoint, i.e., $\ell_1\le \ell_2\le \dots \le  \ell_n$. 
\begin{defn}
\label{def:S}
    Consider any iteration $\i\ge 0$. 
    \begin{itemize}
        \item Let $a_\i$ and $b_\i$ denote the  first and second $\i$-big intervals, respectively. 
    \item $A_\i:=\{1,2,\dots, a_\i-1\}$ is the maximal prefix of $[n]$  that does not contain any $\i$-big interval.
    \item  $B_\i:=\{a_\i+1,\dots, b_\i-1\}$ is the segment of $[n]$ between the  first and second $\i$-big intervals.
    \end{itemize}
    We now define the almost-prefix query set $\t_\i$ as follows:
\begin{enumerate}
    \item If $c(A_\i)> y^\i$ then $S_\i$ is the maximal prefix of $A_\i$ having cost at most $2y^\i$. Here, $\t_\i\subseteq A_\i$.
    \item    If $c(A_\i)\le  y^\i$ then $S_\i$ is the maximal prefix of $A_\i\cup B_\i$ having cost at most $y^\i$. Here, $\t_\i\supseteq A_\i$.
\end{enumerate}
    \end{defn}

\begin{figure}[h!]
    \centering
      \caption{Illustration of Definition \ref{def:S}.  }
\includegraphics[width=\textwidth]{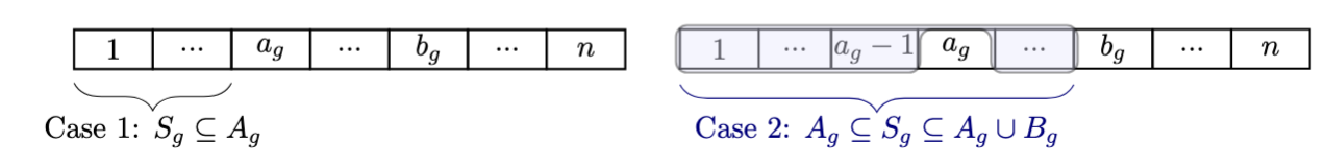}
\end{figure}
    
\paragraph{\smqi algorithm.} This involves replacing the ``prefix set'' $T_\i $ in Step~\ref{alg:step-a} of the \smq algorithm (Algorithm~\ref{pol:gengreedy}) by the almost-prefix set $\t_\i$ defined above. The  other steps  remain the same as in Algorithm~\ref{pol:gengreedy}.  The stopping criterion also changes: we will perform queries in the order of $\pi$ until either \eqref{eq:smqi-rule-1} or \eqref{eq:smqi-rule-2} applies.  We start with a useful lemma showing that the almost-prefix sets $\t_\i$ are nested (as was the case for the sets $T_\i$). We note that this lemma is just needed to obtain a tighter constant factor. 

\begin{lem}\label{lem:smqi-nested}
    For each $\i\ge 0$, we have $\t_\i \subseteq \t_{\i+1}$. 
\end{lem}
\begin{proof}  First, suppose that $c(A_\i)> y^\i$ (corresponds to case 1 in Definition~\ref{def:S}). Then, set $\t_\i\subseteq A_\i$ is a prefix of cost at most $2\base^\i$.      
   Clearly, the first $(\i+1)$-big interval $a_{\i+1}\ge a_\i$, so $A_\i\subseteq A_{\i+1}$.  
   \begin{itemize}
    \item  
   If $c(A_{\i+1})\le \base^{\i+1}$ then $\t_{\i+1}\supseteq A_{\i+1}\supseteq A_\i \supseteq \t_\i$. 
   \item If $c(A_{\i+1})>\base^{\i+1}$ then $\t_{\i+1}$ is the maximal prefix of $A_{\i+1}\supseteq A_\i$ of cost at most $2\base^{\i+1}> 2\base^\i$: so  we must have $\t_\i\subseteq \t_{\i+1}$.
   \end{itemize}

    Now, suppose that $c(A_\i)\le y^\i$ (corresponds to case 2 in  Definition~\ref{def:S}). Here, $\t_\i\supseteq A_\i$ and is a prefix of $A_\i\cup B_\i$ with $c(\t_\i)\le \base^\i$. 
    If $a_\i$ is also $(\i+1)$-big then $A_{\i+1}=A_\i$ and $c(A_{\i+1})\le y^{\i+1}$: so $\t_{\i+1}$ corresponds to case 2 in  Definition~\ref{def:S}. Also, the second $(\i+1)$-big interval $b_{\i+1}\ge b_\i$: so, $B_{\i+1}\supseteq B_\i$. Hence,  $\t_{\i+1}$ is the maximal prefix of $A_\i\cup B_{\i+1}\supseteq A_\i\cup B_\i$ of cost at most $y^{\i+1}$. So, we must have $\t_{\i+1}\supseteq \t_\i$.  
 If $a_\i$ is not $(\i+1)$-big then we have $\base^\i < c(a_\i)\le \base^{\i+1}$ and the first $(\i+1)$-big interval  $a_{\i+1}\ge b_\i$, i.e., $A_{\i+1}\supseteq  A_\i\cup \{a_\i\} \cup B_\i  $. 
    \begin{itemize}
    \item  
If $c(A_{\i+1})>\base^{\i+1}$ then $\t_{\i+1}$ corresponds to case 1 in    Definition~\ref{def:S}.   Consider the prefix $\{a_\i\}\cup \t_\i$: it has cost at most $\base^{\i+1} + \base^\i< 2\, \base^{\i+1}$. So, we must have $\t_{\i+1}\supseteq \{a_\i\}\cup \t_\i$.

\item If $c(A_{\i+1})\le \base^{\i+1}$ then $\t_{\i+1}$ corresponds to case 2 in    Definition~\ref{def:S}. Here, $\t_{\i+1}\supseteq A_{\i+1} \supseteq A_\i\cup \{a_\i\} \cup B_\i \supseteq \t_\i$.
   \end{itemize} 
   In all cases, we have $\t_\i \subseteq \t_{\i+1}$. 
\end{proof}

The rest of the analysis combines ideas from the non-uniform cost \smq and the uniform cost \smqi. Let $\pi$ denote our (non-adaptive) policy and $\sigma$ the optimal adaptive policy. We re-use the terms from Definition~\ref{def:ai}: 
  \[ o_\i := \Pr[ \sigma \text{ does not finish by cost }\base^\i] . \]
    \[ v_\i := \Pr[\pi \text{ does not finish by iteration }\i] .\]
    \[\sigma_\i \,\,  \text{ is the optimal policy truncated at cost }\base^\i\]
    \[\pi_\i \,\,  \text{ is our policy truncated at the end of iteration  } \i.\]

\begin{lem}
\label{lem:smqi-cost-i}The cost of our policy until the end of iteration $\i$  is
   \[c(\pi_\i)  \leq (1 + \epsilon) \left(2 + \dfrac{\base}{\base - 1}  \right)\base^\i. \]
\end{lem}
\begin{proof}
We handle separately the costs of intervals queried in the (modified) Step \ref{alg:step-a}   and Step~\ref{alg:step-b} of Algorithm~\ref{pol:gengreedy}. By Lemma~\ref{lem:smqi-nested}, we have $\cup_{k=0}^\i \t_k =\t_\i$. So,  the total cost incurred in the modified Step~\ref{alg:step-a} of the first $\i$ iterations is $c(\t_\i)\le 2\cdot \base^\i$. The total cost due to Step~\ref{alg:step-b} is exactly as in Lemma~\ref{lem:grdup}, which is at most $ (1 + \epsilon) \cdot \dfrac{\base^{\i+1}}{\base - 1}$.   The lemma now follows. 
\end{proof}
Lemma~\ref{lem:optlb} continues to hold here as well; so:
\begin{equation}
    \label{eq:lem-opt-lb}
    \sum_{\i\ge 0} {\base^\i o_\i} \leq  \frac{\base }{\base-1} c_{exp}(\sigma) - \frac1{y-1}. 
\end{equation}

The key step is the analogue of Lemma~\ref{lem:aioi}, which we prove in the next subsection.
 \begin{lem} For all $\i \ge 0$, we have $v_\i \le o_\i$. 
 \label{lem:smqi-comp-prob}
 \end{lem}

We can now prove the main result. 
\begin{thm}
      There is a $( 4+ 2\sqrt{3} + \epsilon)$-approximation for \smqi  with general costs.  

\end{thm}
\begin{proof} By Lemma~\ref{lem:smqi-cost-i}, we have   $c_{exp}(\pi)  \leq (1 + \epsilon)  \left(2+ \frac{\base}{\base-1}  \right) \sum_{\i \ge 1}{{\base}^\i(v_{\i-1} - v_\i)}$. Exactly as in the proof of Theorem~\ref{thm:gencost}, using \eqref{eq:lem-opt-lb}, we get
$\sum_{\i \ge 1} {{\base}^\i(v_{\i-1} - v_\i)} \le y\cdot c_{exp}(\sigma)$. Therefore, $c_{exp}(\pi)  \leq (1 + \epsilon)  y\cdot \left(2+ \frac{\base}{\base-1}   \right) \cdot  c_{exp}(\sigma)$. Now, optimizing for $y$, we obtain the stated approximation ratio.
\end{proof}

\subsection{Proof of Lemma~\ref{lem:smqi-comp-prob}}

Fix any iteration $\i$. As in the unit-cost \smqi proof, we define the following events.

$$ \cA_1 \,:\, \text{    our policy $\pi$   finishes within $\i$ iterations due to~\eqref{eq:smqi-rule-1}.}$$
$$ \cA_2 \,:\, \text{    our policy $\pi$   finishes within $\i$ iterations due to~\eqref{eq:smqi-rule-2}.}$$
$$ \cO_1 \,:\, \text{    optimal policy $\sigma$   finishes by  cost $y^\i$  due to~\eqref{eq:smqi-rule-1}.}$$
$$ \cO_2 \,:\, \text{    optimal policy $\sigma$   finishes  by  cost $y^\i$  due to~\eqref{eq:smqi-rule-2}.}$$

Clearly, $1-v_\i = \Pr[\cA_1\vee \cA_2]$ and $1-o_\i = \Pr[\cO_1\vee \cO_2]$. 

\paragraph{Handling the old stopping criterion.}
Let $L$ denote the smallest un-queried left-endpoint at the end of iteration $\i$ in $\pi$.  Note that $L$ is a deterministic value. Let  
$\cG$ be the event that $X_j > L+\delta$ for all intervals $j$  queried by $\pi_\i$.  Note that $\cG=\neg \cA_1$, i.e., criterion~\eqref{eq:smqi-rule-1} does {\em not} apply by the end of iteration $\i$. 
Recall that threshold $\theta_\i$ (Step~\ref{alg:step-theta} in Algorithm \ref{pol:gengreedy})  is $\delta$ more that the smallest un-queried left-endpoint in that step. Clearly,  $\theta_\i \le L+\delta$ in iteration $\i$.

Now consider the truncated optimal policy. Let $L(\sigma_\i)$ be its smallest un-queried left-endpoint; this is a random value as $\sigma_\i$ is an adaptive policy. We claim that 
\begin{equation}
    \label{eq:left-dominance}
    L+\delta \,\, \ge  \,\,\theta_\i \,\, \ge  \,\,\min_{j\in N\setminus \t_\i} \ell_j  +\delta  \,\,\ge  \,\,L(\sigma_\i) +\delta.
\end{equation}
Above,  the second inequality uses the fact that $\t_\i$ is queried before Step~\ref{alg:step-theta}.  To see the last  inequality, note that  $\sigma_\i$ cannot query any $\i$-big interval: so  $L(\sigma_\i)\le \ell_{a_\i}$.  We have two cases depending on the definition of $\t_\i$:
\begin{itemize}
    \item If $\t_\i\supseteq A_\i$ then clearly $\min_{j\in N\setminus \t_\i} \ell_j = \ell_{a_\i} \ge L(\sigma_\i)$. 
    \item If $\t_\i\subseteq A_\i$ then we must have $c(A_\i)>y^\i$, which means that $\t_\i$ contains the maximal prefix of cost at most $y^\i$. Again, this implies $\min_{j\in N\setminus \t_\i} \ell_j \ge  L(\sigma_\i)$.
\end{itemize}
This proves \eqref{eq:left-dominance}. 

Now, let $\cG^*$ be the that event that $X_j > \theta_\i$ for all intervals $j$ in   $\sigma_\i$. Using \eqref{eq:left-dominance} it follows that $\cO_1\subseteq \neg \cG^*$.  We now obtain:
\begin{equation}\label{eq:smqi-prob-g*}
\Pr[\cG^*] = \Pr\left[\min_{j \in \sigma_\i} X_j  > \theta_\i\right] \ge   \Pr\left[\min_{j \in \pi_\i} X_j  > \theta_\i\right]\ge   \Pr\left[\min_{j \in \pi_\i} X_j  > L+\delta \right] = \Pr[\cG].
\end{equation}
 The first inequality follows from the proof of Lemma~\ref{lem:aioi}: see \eqref{eq:opt-theta} - \eqref{eq:knap}. The second inequality above uses $\theta_\i\le L+\delta$.   

\paragraph{Handling the new stopping criterion.} Let $\cG_A$ be the event that $X_j>L+\delta$ for {\em all} r.v.s $j\in N$. Similarly, let $\cG^*_A$ be the event that $X_j>\theta_\i$ for {\em  all}  $j\in N$. Clearly,
$\Pr\left[\cA_2\, |\, \cG\right]  \,=\, \Pr\left[\cA_2\, |\,  \cG_A\right]$  and $ \Pr\left[\cO_2\, |\, \cG^*\right]  \,=\, \Pr\left[\cO_2\, |\,  \cG^*_A\right]$.  
  We will now prove that
 \begin{equation}\label{eq:smqi-prob-new-stop}
\Pr\left[\cA_2\, |\, \cG\right]  \,=\, \Pr\left[\cA_2\, |\, \cG_A\right] \,\ge\,  \Pr\left[\cO_2\, | \, \cG^*_A\right] \,=\, \Pr\left[\cO_2\, |\, \cG^*\right] .
\end{equation}
Using \eqref{eq:smqi-prob-g*}, exactly as in the proof of Theorem~\ref{thm:smqi-unit}, this  implies 
$\Pr[\cA_1 \vee \cA_2]  \ge  \Pr[\cO_1 \vee \cO_2]$. We repeat the argument below for completeness.
\begin{align*}
    \Pr[\cA_1 \vee \cA_2] &=     \Pr[\cA_1] +     \Pr[\cA_2 \wedge \neg \cA_1] = \Pr[\neg \cG] + \Pr[\cA_2 \wedge  \cG]\\
    & = \Pr[\neg \cG] + \Pr[\cA_2 |\cG]\cdot \Pr[\cG] = 1 - \left( 1 -  \Pr[\cA_2 |\cG]\right) \cdot \Pr[\cG]\\
    & \ge 1 - \left( 1 -  \Pr[\cO_2 |\cG^*]\right) \cdot \Pr[\cG^*]\qquad \text{ by \eqref{eq:smqi-prob-g*} and \eqref{eq:smqi-prob-new-stop}}\\
    &= \Pr[\neg \cG^*] + \Pr[\cO_2 \wedge  \cG^*]\\
    & \ge \Pr[\cO_1] +     \Pr[\cO_2 \wedge \neg \cO_1] \qquad \text{ using  }\cO_1\subseteq \neg \cG^* \\
    & =    \Pr[\cO_1 \vee \cO_2].
\end{align*}

This proves Lemma~\ref{lem:smqi-comp-prob}.

\paragraph{Proving \eqref{eq:smqi-prob-new-stop}.}  The key property here is the following.
\begin{lem}\label{lem:smqi-stop-S}
    If $\sigma$ finishes due to criterion~\eqref{eq:smqi-rule-2} and identifies $i^*$ by cost  $y^\i$  then $P_{i^*}\subseteq \t_\i$.
\end{lem}
\begin{proof}
 Clearly $c(P_{i^*})\le y^\i$ as $\sigma$ finishes by cost $y^\i$.   Recall that $P_{i^*}=\{ j\in N\setminus i^* : \ell_j < r_{i^*} - \delta\}$ is an  almost-prefix set. We consider two cases:
    \begin{itemize}
        \item $P_{i^*}$ is itself a prefix. In this case, the first $\i$-big interval $a_\i$ must occur after $P_{i^*}$, i.e., $P_{i^*}\subseteq A_\i$. By the definition of $\t_\i$, it either contains $A_\i\supseteq P_{i^*}$ or $\t_\i$ is a maximal prefix of cost at most $2y^\i$ (which also contains $P_{i^*}$). 
        \item $P_{i^*}$ is not  a prefix, but $P_{i^*} \cup i^*$ is a prefix. 
        
        If $i^*$ is not $\i$-big then $P_{i^*} \cup i^*$ has cost at most $2 y^\i$ and is a subset of $A_\i$ (as $P_{i^*} \cup i^*$  cannot contain any $\i$-big interval). So $P_{i^*} \cup i^*$   is contained in the maximal prefix of $A_\i$ having cost at most $2y^\i$, which is always contained in $\t_\i$.
    
    If $i^*$ is $\i$-big then $i^*=a_\i$ the first $\i$-big interval (otherwise $P_{i^*}$ would contain some $\i$-big interval). This also means that   $c(A_\i)\le y^\i$: so $\t_\i$ is the maximal prefix in $[n]\setminus \{a_\i\}$ of cost at most $y^\i$. Clearly, we must then have $P_{i^*}\subseteq \t_\i$.
    \end{itemize}
\end{proof}

Let $R=\{h\in N: P_h \subseteq \t_\i\}$. By Lemma~\ref{lem:smqi-stop-S} it follows that if event $\cO_2$ occurs then  $i^*\in R$. Hence, $\cO_2$ 
is a subset of the event
$$\cE := \bigvee_{h\in R}  \left( \wedge_{j\in P_h} (X_j > r_h-\delta)\right) .$$

Moreover, our  policy $\pi_\i$ queries all the r.v.s in $\t_\i$. So, for all $h\in R$, the r.v.s in $P_h\subseteq \t_\i$ are queried by $\pi_\i$. Hence,  event $\cA_2$  
contains event $\cE$.

Recall that the event $\cG_A$ (resp. $\cG^*_A$) in policy $\pi$ (resp. $\sigma$) means that every  r.v. is more than $L+\delta$ (resp. $\theta$). Also, $\theta \le L+\delta$, which means 
$$\Pr[X_j > t | X_j >L+\delta] \ge \Pr[X_j > t | X_j >\theta], \quad \forall t\in \mathbb{R} , \forall j\in N. $$
In other words, for any $j\in N$, r.v. $X_j$ conditioned on $\cG_A$ stochastically dominates $X_j$ conditioned on $\cG^*_A$.  Using Lemma~\ref{lem:event-sd} (which deals with the same event $\cE$) with $Y_j = X_j | \cG_A$ and $Z_j = X_j | \cG^*_A$, we obtain $\Pr[\cE  | \cG_A] \ge \Pr[\cE | \cG^*_A]$, which proves \eqref{eq:smqi-prob-new-stop}.

\bibliographystyle{alpha}
\bibliography{ref}

\appendix
\section{Multiplicative Precision}\label{app:multiplicative}
Given an instance with non-negative r.v.s $\{X_i\}_{i=1}^n$ and multiplicative precision $\alpha \ge 1$, consider a new instance of \smq with r.v.s $\{X'_i := \ln(X_i)\}_{i=1}^n$ and additive precision $\delta:= \ln \alpha$. Note that 
$$\mv'=\min_{i=1}^n X'_i=\min_{i=1}^n \ln(X_i) = \ln\left(\min_{i=1}^n X_i\right)=\ln (\mv).$$
An  $\alpha$-approximately minimum value $W$ for the original instance satisfies $\mv\le W\le \alpha\cdot \mv$, where $\mv=\min_{i=1}^n X_i$. Then,  $\vv=\ln(W)$ satisfies
$\mv' = \ln (\mv) \le \vv \le \ln(\mv) + \ln \alpha = \mv'+\delta$, i.e., $\vv$ is a $\delta$-minimum value for the new instance. Similarly, if $\vv$ is a $\delta$-minimum value for the new instance then $W:=e^{\vv}$ is an $\alpha$-approximately minimum value  for the original instance.

\section{Bad Example for Competitive Ratio}\label{app:bad-cr}
We provide an example that rules out any reasonable {\em competitive ratio} bound for \smq and \smqi with precision $\delta>0$. This is in sharp contrast to the corresponding problem with exact precision ($\delta=0$)  for which a constant competitive ratio is known~\cite{kahan1991model}.  We note that results in the online setting assume open intervals, which in our setting (with discrete r.v.s)  corresponds to all left-endpoints being distinct.\footnote{Alternatively, our example can be modified into one with open intervals where the competitive ratio is still $\tilde\Omega(n)$.}  The benchmark in the online setting is the {\em hindsight optimum}, which is the minimum number (or cost) of queries that are needed to verify a $\delta$-minimum value  {\em conditioned} on the realizations $\{x_i\}_{i=1}^n$ of the r.v.s.  

Consider an instance with $n$ r.v.s with $\Pr[X_i=i]=p:=\frac{\ln n}n$ and $\Pr[X_i=n^2] = 1-p$ for all $i\in [n]$. All costs are unit and the precision $\delta=n$. We refer to the values $\{1,2,\dots, n\}$ as {\em low} values: note that any low value is a $\delta$-minimum value for this instance.  

We first consider the hindsight optimum. If any of the $n$
 r.v.s (say $k$) realizes to  a low value then verifying the $\delta$-minimum value just requires querying $k$, which has cost $1$. On the other hand, the probability that none of the  $n$ r.v.s   realizes to  a low value is $\left(1-p\right)^n\le e^{-pn}=\frac1n$: in this case the optimal verification cost is $n$ (querying all r.v.s). So the expected optimal cost is at most $2$.  

Now, consider any \smq policy: this does not know the realizations before querying. The only way to stop querying is (1) when some low value is observed, or (2) all $n$ r.v.s have been queried. 
The probability that the $i^{th}$ r.v. is queried is exactly $(1-p)^{i-1}$, which corresponds  no low realization among the previous $i-1$ r.v.s. So, the expected cost of any policy is:
$$\sum_{i=1}^n (1-p)^{i-1} = \sum_{i=0}^\infty (1-p)^i - \sum_{i=n}^\infty (1-p)^i = \frac1p -  \frac1p  (1-p)^n \ge \frac1p (1-e^{-pn}),$$
where the second equality  uses $\sum_{i=0}^\infty (1-p)^i=\frac1p$. Using $p=\frac{\ln n}{n}$, the expected cost is at least $\frac{n}{\ln n}(1-\frac1n)$.

Hence the competitive ratio for \smq is $\Omega(\frac{n}{\ln n})$.

\section{Adaptivity Gap for \smq }
\label{apx:na}

The instance has three intervals with $X_1 \in \{0, 3, \infty\} $, $X_2 \in \{1, \infty\} $, $X_3 \in \{2,  \infty\} $ and $\delta = 1$. Let  $\Pr(X_1 = 0) = \frac{1}{3}, \Pr(X_1 = 3) = \frac{1}{3}, \Pr(X_1 = \infty) = \frac{1}{3}, \Pr(X_2 = 1  ) = \epsilon,  \Pr(X_2 = \infty  ) = 1- \epsilon,  \Pr(X_3 = 2  ) = 1 - \epsilon, \Pr(X_3 = \infty) = \epsilon$. Recall that the  adaptive policy has cost at most $\frac{5+\epsilon}{3}$.

 We consider the cost of all possible non-adapative policies:
\begin{enumerate}
    \item 
  The cost of policy $\{1,2,3\}$ is,
  \begin{align*}
    NA & \geq 1 + \frac{2}{3} +  \frac{1 - \epsilon}{3} = \frac{6 -  \epsilon}{3} .
\end{align*}
    We query $X_1$ w.p. 1,  $X_2$  w.p. $2/3$ and  $X_3$  w.p. ${(1-\epsilon)}/{3}$.

\item 

The cost of policy $\{1,3,2\}$ is,
\begin{align*}
    NA & \geq 1 + \frac{2}{3} + \frac{2\epsilon}{3}  = \frac{5 + 2 \epsilon}{3} .
\end{align*}
    We query  $X_1$  w.p. 1,  $X_3$  w.p. $2/3$ and  $X_2$  w.p. ${2\epsilon}/{3}$.

\item 

The cost of policy $\{2,1,3\}$ is,
\begin{align*}
    NA & \geq 2 - \epsilon + \frac{1-\epsilon}{3} = \frac{7- 4\epsilon}{3}. 
\end{align*}

    We query  $X_2$   w.p. 1,  $X_1$  w.p. $1- \epsilon$ and  $X_3$  w.p. $(1 - \epsilon)/3$.

\item 

The cost of policy $\{2,3,1\}$ is,
\begin{align*}
    NA & \geq 2 - \epsilon + 1 - \epsilon = 3 - 2\epsilon .  
\end{align*}
 
    We query $X_2$ w.p. 1,  $X_3$    w.p. $1- \epsilon$ and  $X_1$ w.p. $1 - \epsilon$.

\item  The remaining non-adaptive policies  start with $3$. Any such policy costs at least $2$ because even if $X_3=2$ (its lowest value) we cannot stop.

\end{enumerate}

So, the optimal non-adaptive value is $\frac13 \cdot \min\left\{ 6-\epsilon , 5+2\epsilon , 7-4\epsilon, 9-6\epsilon, 6\right\}$. Setting   $\epsilon=\frac13$, the non-adaptive optimum is $\frac{17}9$, whereas the adaptive optimum is $\frac{16}{9}$.  

Hence, the  adaptivity gap for \smq is at least $\frac{17}{16}$.

\begin{rk}
        We note that if we allow $\Pr(X_2=1 ) = \epsilon_2  < \epsilon_3 = \Pr(X_3 = \infty) $ then we can achieve a ratio that is equal to $\frac{12}{11}$ as $\epsilon_2 \to 0$ and $\epsilon_3 = 0.5$.
       
\end{rk}

\section{Fixed Threshold Problem }\label{app:fixed-t}
Here, we prove Proposition~\ref{prop-fixed-t}. We proceed by induction on the budget $k$. For any set $S$ of r.v.s and budget $k$, let
\[V(S,k) := \max_{\mathcal{A} \subseteq S, c(\mathcal{A}) \leq k}\quad \Pr_{\mathcal{A},X}\left[\min_{j \in \mathcal{A}} X_j  \le \theta\right],\]
denote the maximum success probability over adaptive policies (having  cost at most  $k$). Similarly,   
\[F(S,k) = \max_{T \subseteq S: c(T)\le k}\quad  \Pr_X\left[ \min_{j \in T} X_j \le  \theta\right]\] 
be the maximum over non-adaptive policies. 
We will show that $V(S,k)=F(S,k)$, which would prove   Proposition~\ref{prop-fixed-t}. It suffices to show $V(S,k) \leq F(S,k)$. (Clearly,  $V(S,k) \geq F(S,k)$ as  adaptive policies capture all non-adaptive policies.)

The base case ($k=1$) is trivial because any policy is non-adaptive (it selects a single r.v.).


\begin{figure}[H]
    \centering
\begin{forest}
  [$X_a$
    [$w_1$ [$\mathcal{A}_{w_1}$, roof]
    ]
    [$w_2$
        [$\mathcal{A}_{w_2}$, roof]
    ]
    [... 
    ]
    [$w_i$ 
    ]
    [... 
    ]
  ]
\end{forest}
    \caption{Adaptive policy ${\cal A}$ for the fixed threshold problem.\label{fig:policy_tree}}
    \end{figure}
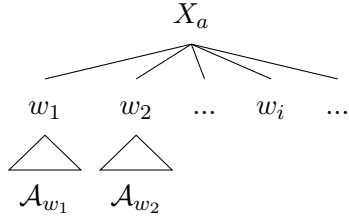

For the inductive step,  we fix some budget $\ell$ and want to show $V(S,\ell ) \leq F(S,\ell)$.  For any policy $\pi$, we will use $prob(\pi):=\Pr[\min_{i\in \pi} X_i\le \theta]$ to denote its success probability.  Let  $\mathcal{A}$
denote the optimal adaptive policy, which has $prob({\cal A}) = V(S,\ell )$. Let $a\in S$ denote the first query in policy  $\mathcal{A}$. Let $T^+$ (resp. $T^-$) represent all realizations of $X_a$ that are at most (resp. more than)  threshold $\theta$. For any realization $w$ of $X_a$, let $\mathcal{A}_w$ denote the rest of  policy ${\cal A}$ {\em conditioned} on $X_a=w$; note that the cost $c(\mathcal{A}_w)\le \ell - c_a$ because policy ${\cal A}$ always has cost at most $\ell$. See Figure~\ref{fig:policy_tree}.  Below, we use $p_a:=\Pr[X_a \le \theta]=\sum_{w\in T^+} \Pr[X_a=w]$; so $\sum_{w\in T^-} \Pr[X_a=w]=1-p_a$. We now have:
\begin{align} 
    V(S,\ell) & =  prob({\cal A}) =  p_a + \sum_{w \in T^-} \Pr[X_a = w] \cdot prob({\cal A}_w)  \notag\\
    & \le  p_a + \sum_{w \in T^-} \Pr[ X_a = w] \cdot V(S \setminus a, \ell - c_a) =  p_a + (1- p_a)\cdot V(S \setminus a, \ell - c_a)   \label{eq:Tplus}\\
     & \leq    p_a + (1 - p_a)\cdot F(S \setminus a, \ell - c_a)  \leq F(S, \ell) \label{eq:indstep}
     \end{align}
The inequality in \eqref{eq:Tplus} uses the fact that each ${\cal A}_w$ is a feasible adaptive policy for the smaller instance on r.v.s $S\setminus a$ and budget $\ell-c_a$. The first inequality in  \eqref{eq:indstep} is by induction. The second inequality in  \eqref{eq:indstep} is by the following observation. Let $T\subseteq S\setminus a$  be an optimal non-adaptive policy for the instance $F(S \setminus a, \ell - c_a)$; then $T\cup a$ is a feasible non-adaptive policy for the instance $F(S,\ell)$ with success probability $p_a + (1-p_a)\cdot prob(T)= p_a + (1 - p_a)\cdot F(S \setminus a, \ell - c_a) $.

\section{The  Knapsack  Subroutine \texorpdfstring{\eqref{eq:algknap}}{(algknap)}}
\label{app:knapsack}

We now provide a bi-criteria approximation algorithm for the knapsack instance \eqref{eq:algknap}. 
\begin{thm}
\label{thm:knapsack}
    Given discrete random variables  $\{X_i\}_{i=1}^n$ with costs $\{c_i\}_{i=1}^n$, budget $d$ and  threshold $\theta \in \mathbb{R}$, there is an $n^{O(1/\epsilon)}$  time algorithm that finds $T \subseteq N$ such that  $ \Pr\left[ \min_{j \in T} X_j   > \theta \right]\le  p^*$ and $c(T) \leq (1+\epsilon)d$, for any $\epsilon>0$. Here, 
    \[
      p^* = \min_{T \subseteq N } \left\{ \Pr\left[ \min_{j \in T} X_j  > \theta \right] : c(T) \leq d \right\}. \tag{*} \label{KP}
    \]
\end{thm}

\begin{proof}

First, we re-write \eqref{KP} as a weighted knapsack problem, using  

    \[  \Pr\left[ \min_{j \in T}  X_j  > \theta \right] = \prod_{j \in T} \underbrace{\Pr[X_j > \theta ]}_{q_j} . \]
 Then, we take an inverse (which converts the min objective to max) and the logarithm (which makes  the objective linear). 

     \[ \log(\frac1{p^*}) =  \max_{T \subseteq N }\left\{\sum_{j \in T}{\log\left(\frac1{q_j} \right)} : c(T) \leq d \right\}. 
    \]

    Let $r_j:=\log(1/q_j)\ge 0$ be the ``reward'' of each item. Then the above problem is just the usual knapsack problem. We now provide a bicriteria approximation algorithm using standard enumeration techniques combined with a greedy algorithm. (We provide the full proof because we did not find a  reference to the precise  bi-criteria guarantee that is needed here.)

\paragraph{Algorithm: Bicriteria Knapsack}
    \begin{enumerate}
    \label{alg:knap}
        \item Order items greedily such that  $ \frac{r_1}{c_1} \ge \frac{r_2}{c_2} \ge ... \ge \frac{r_n}{c_n}$.
        \item Let $B = \{j\in [n]: c_j > \epsilon d \}$ be the set of ``large'' items.
        \item For each $S \subseteq B$ with $c(S) \leq d$:
        \begin{enumerate}
            \item Initialize solution $T_S\gets  S$.
            \item Add items from $[n]\setminus B$ to $T_S$ in the greedy order  until $c(T_S)$ exceeds $d$ for the first time.
        \end{enumerate}
        \item Return the best solution $T_S$ obtained above.
    \end{enumerate}

We first show that the runtime is $n^{O(1/\epsilon)}$. The key observation is that the number of distinct subsets $S\subseteq B$ considered in Step~3 is at most $n^{1/\epsilon}$: this is because each item in $B$ has cost more that $\epsilon d$.  

Let $T$ be  the solution found at the end of the algorithm. It is clear that  $d \le c(T) \le d+c_{max}$ where $c_{max}=\max_{i\in [n]\setminus B}  c_i$. Note that $c_{max}\le \epsilon d$ by definition of the large items $B$. So,  $c(T) \leq (1+\epsilon)d$. 

We now bound the total ``reward'' $\sum_{i\in T} r_i$. We will use the following well-known fact about the greedy algorithm for knapsack.
\begin{quote}
    Consider any knapsack instance with items $I$, rewards $\{r_i\}_{i\in I}$, costs $\{c_i\}_{i\in I}$ and budget $d'$. Let $G\subseteq I$ be the ``greedy'' solution obtained by including items in decreasing order of the ratio $\frac{r_i}{c_i}$ until $c(G)>d'$ for the first time. Then $r(G)\ge \max_{K\subseteq I: c(K)\le d'} \sum_{i\in K}r_i$. 
\end{quote}

Let $T^*$ be the optimal solution to \eqref{KP}. Then, $S=T^*\cap B$ will be one of the choices for $S$ considered in Step~3. Moreover, the set $T_S\setminus S$ added in Step~3b is precisely  the greedy solution for the knapsack instance with  items $I=[n]\setminus B$ and budget $d'=d-c(S)=d-c(T^*\cap B)$. The above fact implies that $r(T_S\setminus S)\ge r(T^*\setminus B)$ because $T^*\setminus B$ is a feasible solution to this knapsack instance.
It follows that $r(T_S)=r(S)+r(T_S\setminus S)\ge r(T^*)=\log(\frac1{p^*})$. 
\end{proof}

\end{document}